\documentclass[letterpaper,12pt]{article}

\usepackage{amsmath} 
\usepackage{amsfonts} 
\usepackage{amssymb}
\usepackage{amsthm} 
\usepackage{eucal}
\usepackage{mathrsfs}
\usepackage{enumerate}
\usepackage{calc}
\usepackage{float}
\usepackage{tabularx}
\usepackage{multirow}
\usepackage{graphicx}

\renewcommand{\vec}[1]{\mathbf{#1}}

\newcommand{\bpi}{\boldsymbol{\pi}}

\newcommand{\R}{\mathbb{R}}

\newcommand{\N}{\mathbb{N}}

\newcommand{\e}{\varepsilon}

\newcommand{\cj}[1]{{#1}^\ast}
\newcommand{\qam}[1]{{#1}^\mathfrak{s}}
\newcommand{\qsum}[1]{{#1}^\mathfrak{u}}

\newcommand{\cball}[2]{\overline{\mathfrak{B}}\big(#1,#2\big)}

\newcommand{\SW}{\mathbf{SW}}
\newcommand{\algn}[1]{\bigl\langle#1\bigr\rangle}
\newcommand{\gmcl}{\Gamma_{\textrm{CL}}(\Sigma)}
\newcommand{\DM}[3]{\mathsf{EM}^{#1}(#2,#3)}
\newcommand{\DQ}[4]{\mathsf{EQ}^{#1}(#2,#3,#4)}
\newcommand{\GQ}[4]{\mathsf{GQ}^{#1}(#2,#3,#4)}
\newcommand{\GM}[4]{\mathsf{GM}^{#1}(#2,#3,#4)}
\newcommand{\LQ}[4]{\mathsf{LQ}^{#1}(#2,#3,#4)}
\newcommand{\LM}[4]{\mathsf{LM}^{#1}(#2,#3,#4)}
\newcommand{\AQ}[2]{\mathsf{AQ}^{#1}(#2)}
\newcommand{\AM}[2]{\mathsf{AM}^{#1}(#2)}
\newcommand{\prf}[1]{\hat{#1}}

\newcommand{\sQ}{Q}
\newcommand{\sD}{D}
\newcommand{\sS}{S}
\newcommand{\sH}{H}

\newcommand{\mH}{\mathbf{H}}
\newcommand{\mD}{\mathbf{D}}

\newcommand{\norm}[1]{\left\Vert#1\right\Vert}
\newcommand{\abs}[1]{\left\vert#1\right\vert}

\newtheorem{thm}{Theorem}[section]
\newtheorem{cor}[thm]{Corollary}
\newtheorem{lem}[thm]{Lemma}
\newtheorem{prop}[thm]{Proposition}
\theoremstyle{definition}
\newtheorem{defn}[thm]{Definition}
\newtheorem{exmp}[thm]{Example}
\newtheorem{rem}[thm]{Remark}

\numberwithin{thm}{section}
\newcounter{tmpthm}
\newcounter{tmpsec}

\newcounter{lemselfsimn}
\newcounter{lemselfsims}

\newcounter{thmeditdistqmn}
\newcounter{thmeditdistqms}

\newcounter{corglobdistn}
\newcounter{corglobdists}

\title{Geometric Aspects of Biological Sequence Comparison}
\author{Aleksandar Stojmirovi\'{c} and Yi-Kuo Yu\\[.5cm] \small National Center for Biotechnology Information,\\ \small National Library of Medicine, National Institutes of Health,\\ \small Bethesda, MD 20894, United States}

\begin{document}

\maketitle

\begin{abstract}
We introduce a geometric framework suitable for studying the relationships among biological sequences. In contrast to previous works, our formulation allows asymmetric distances (quasi-metrics), originating from uneven weighting of strings, which may induce non-trivial partial orders on sets of biosequences. The distances considered are more general than traditional generalized string edit distances. In particular, our framework enables non-trivial conversion between sequence similarities, both local and global, and distances. Our constructions apply to a wide class of scoring schemes and require much less restrictive gap penalties than the ones regularly used. Numerous examples are provided to illustrate the concepts introduced and their potential applications.
\end{abstract}

\section{Introduction}\label{sec:intro}

Biological macromolecules such as DNA, RNA and proteins play an essential role in all living organisms. Structurally, they are all chains of residues belonging to a small set of basic molecules and the functional characteristics of each macromolecule are determined by the order and composition of its components. It is therefore not surprising that comparison and alignment of biological sequences is one of the most important contributions from computational biology to modern biosciences.

Typical approaches to biosequence comparison are either distance- \cite{Se74,WSB76} or similarity-based \cite{NW70,SW81}. The distance-based approaches minimize the cost, while those based on similarity maximize the likelihood of transformation of one sequence into another. In both cases the comparison scores for sequences are obtained by extension from scores over alphabets of basic molecules. The algorithms for computation of alignments are based on the dynamic programming technique \cite{BHK59}. Similarity-based methods became widely accepted because the Smith-Waterman algorithm \cite{SW81} allows computation of local alignments, involving only parts of sequences to be compared. Local alignments are highly appropriate in biological context because elements of structure and function are usually restricted to discrete regions of biosequences and hence strong similarity of fragments of two sequences need not extend to similarity of full sequences. Most distance methods have been global in nature and could not be easily adapted for local comparison.

A downside of using local similarities for sequence comparison is that, while their statistics can be characterized \cite{KaA90,KLY05}, no constraints, apart from algorithmic ones, are placed on the form that similarity measures can take. Under such conditions, sets of biosequences with similarity measures cannot be identified with mathematical structures such as metric or normed spaces, which are a natural framework for many computational techniques such as clustering \cite{WXS03} and indexing for similarity search \cite{HjSa03}. In contrast, distance measures on sequences naturally correspond to metrics under some mild restrictions. 

While the duality between global similarities and distances has been recognized very early \cite{SWF81}, it was only recently established independently by Stojmirovi\'c \cite{Stojmirovic2004} and by Spiro and Macura \cite{SM04} that it is possible to transform local sequence similarity scores derived from many popular scoring functions on building blocks of DNA and proteins into distances satisfying the triangle inequality. In the contexts in which they were presented, the results of the above two papers are almost equivalent, however, their perspectives are quite different. Spiro and Macura \cite{SM04} assume symmetric similarity scores and consider the transformation which converts a similarity to a metric, while \cite{Stojmirovic2004} converts similarity into a quasi-metric, a metric without the symmetry axiom. Quasi-metrics naturally correspond to partial orders and are therefore a natural framework for local similarities.

Unlike most existing literature entries, which are concerned with alignment algorithms, this paper aims to show a rigorous connection between similarities and distances that are metrics or quasi-metrics. Our main results are presented in a form that allows transfer to domains that are not necessarily related to classical string transformations and for that reason we use the framework of free semigroups. We define the $\ell^p$-type edit distance, which generalizes the regular edit distance and allows us to consider many more scoring functions on the amino acid alphabet that fail the requirements in \cite{Stojmirovic2004} and \cite{SM04}. Our results also allow for similarities and distances that are asymmetric. In order to have an accurate description of distances generated from similarities, we introduce a novel nomenclature.

Section \ref{sec:prelim} presents the basic definitions. Edit distances and global similarities are discussed in Sections \ref{sec:editdist} and \ref{sec:globalsim}, respectively. Our main result, Theorem \ref{thm:localqm} is presented in Section \ref{sec:localsim} and various kinds of local similarities are discussed as examples. Section \ref{sec:matrices} examines the applicability of our theory to the actual similarity measures used in contemporary computational biology, while Section \ref{sec:future} discusses some possible applications of our results and future directions. We chose to state many of the well-known results formally and to present many examples to enhance readability. The proofs of the established results are either omitted, or, when generalized in our new framework, relegated to Appendix \ref{app:proofs}.

\section{Preliminaries}\label{sec:prelim}

\subsection{Sequences and Free Semigroups}

Recall that the \emph{free monoid} on a nonempty set $\Sigma$, denoted $\Sigma^*$, is the monoid whose elements, called \emph{words} or \emph{strings}, are all finite sequences of zero or more elements from $\Sigma$, with the binary operation of concatenation. The unique sequence of zero letters (empty string), which we shall denote $e$, is the identity element. The \emph{free semigroup} on $\Sigma$, denoted $\Sigma^+$ is the subset of $\Sigma^*$ containing all elements except the identity.

The length of a word $w\in\Sigma^*$, denoted $\abs{w}$, is the number of occurrences of members of $\Sigma$ in it. For $w=\sigma_1\sigma_2\ldots\sigma_n$, where $\sigma_i\in\Sigma$, $\abs{w}=n$ and we set $\abs{e}=0$.

For two words $u,v\in\Sigma^*$, $u$ is a \emph{factor} or \emph{substring} of $v$ if $v=xuy$ for some $x,y\in\Sigma^*$ and $u$ is a \emph{subsequence} or \emph{subword} of $v$ if $v=w^*_1u^*_1w^*_2u^*_2\ldots w^*_nu^*_nw^*_{n+1}$, where $u=u^*_1u^*_2\ldots u^*_n$, $u^*_i\in\Sigma^*$ and $w^*_i\in\Sigma^*$.  For any $x\in\Sigma^*$, we use $\mathfrak{F}(x)$ to denote the set of all factors of $x$.

We call a semigroup (monoid) $(X,\star)$ \emph{free} if it is isomorphic to the free semigroup (monoid) on some set $\Sigma$. The unique set of elements of $X$ mapping to $\Sigma$ under the isomorphism is called the set of \emph{free generators}.

\begin{exmp}
A DNA molecule can be represented as a word in the free semigroup generated by the four-letter nucleotide alphabet $\Sigma=\{\texttt{A},\texttt{T},\texttt{C},\texttt{G}\}$. An RNA molecule is a word in the free semigroup generated by the alphabet $\Sigma=\{\texttt{A},\texttt{U},\texttt{C},\texttt{G}\}$. A protein can be thought of as a word in the free semigroup generated by the standard twenty amino acid alphabet.
\end{exmp}

\begin{exmp}\label{exmp:profiles}
Let $\Sigma$ be a set and denote by $\mathcal{M}(\Sigma)$ the set of all finite measures supported on $\Sigma$. We will call the elements of the free monoid $\mathcal{M}(\Sigma)^*$ \emph{profiles} over $\Sigma^*$. Profiles arise as models of sets of structurally related biological sequences where $\Sigma$ is the nucleotide or amino acid alphabet.
\end{exmp}

As a convention, for any word $u\in\Sigma^*$, the notation $u=u_1u_2\ldots u_n$, where $n=\abs{u}$ shall mean that $u_i\in\Sigma$ while the notation $u=u^*_1u^*_2\ldots u^*_m$ shall imply that $u^*_i\in\Sigma^*$. For all $1\leq k\leq \abs{u}$ we shall use $\bar{u}_k$ to denote the word $u_1u_2\ldots u_k$ and set $\bar{u}_0=e$. 

Let $f:\Sigma\to\R$. The \emph{canonical homomorphic extension of $f$ to the free monoid $\Sigma^*$} is a function $\bar{f}:\Sigma^*\to\R$ such that $\bar{f}(e)=0$ and for all $x\in\Sigma^+$, $\bar{f}(x)=\sum_{i=1}^{\abs{x}} f(x_i)$.

\subsection{Quasi-metrics}\label{subsec:qm}

Quasi-metrics are asymmetric distance functions that generalize metrics and partial orders. With their associated structures, they belong to an area of active research in topology and theoretical computer science \cite{Ku01}. We now produce the standard definitions used in the remainder of this paper.

A \emph{quasi-metric} on a set $X$ is a mapping $d:X\times X\to\R_+$ such that for all $x,y,z\in X$:
\begin{enumerate}[(i)]
\item $d(x,y)=d(y,x)=0\iff x=y$, and
\item $d(x,z) \leq d(x,y)+ d(y,z)$.
\end{enumerate}
The axiom (ii) is known as the \emph{triangle inequality}. If in addition $d$ is symmetric, that is $d(x,y)=d(y,x)$ for all $x,y\in X$, then $d$ is called a \emph{metric}. A  pair $(X,d)$, where $X$ is a set and $d$ a (quasi-) metric, is called a (quasi-) metric space.

For a quasi-metric $d$, its \emph{conjugate} (or \emph{dual}) quasi-metric, denoted $\cj{d}$, is defined on $X\times X$ by $\cj{d}(x,y) = d(y,x)$, and its \emph{associated metric}, denoted $\qam{d}$, by $\qam{d}(x,y)=\max\{d(x,y), d(y,x)\}=d(x,y)\vee \cj{d}(x,y)$. Another frequently used symmetrization of a quasi-metric is the `sum' metric $\qsum{d}$ defined by $\qsum{d}(x,y) = d(x,y)+d(y,x)$.

A (left) open ball of radius $r>0$ centered at $x_0\in X$ with respect to a quasi-metric $d$ is the set $\{x\in X: d(x_0,x)< r\}$ . The collection of all (left) open balls centered at any $x\in X$ with any $r>0$ is a base for a topology on $X$ induced by $d$. This topology is in general $T_0$ but not necessarily $T_1$. For the purpose of this paper, we will call a quasi-metric $d$ \emph{separating} if the induced topology is $T_1$, that is, if $d(x,y)=0$ implies $x=y$ for all $x,y\in X$. Every quasi-metric $d$ also has its \emph{associated partial order}, denoted $\leq_d$, defined by $x\leq_d y \iff d(x,y)=0$.

A quasi-metric $d$ is called a \emph{weightable quasi-metric} \cite{KuVa94} if there exists a function $w: X\to\R_+$, called the \emph{weight function} or simply the \emph{weight}, satisfying for every $x,y \in X$ \[d(x,y) + w(x) = d(y,x) + w(y).\] In this case we call $d$ \emph{weightable} by $w$. A quasi-metric $d$ is \emph{co-weightable} if its conjugate quasi-metric $\cj{d}$ is weightable. The weight function $w$ by which $\cj{d}$ is weightable is called the \emph{co-weight} of $d$ and $d$ is \emph{co-weightable} by $w$. %A triple $(X,d,w)$ where $(X,d)$ is a quasi-metric space and $w$ a function $X\to\R_+$ is called a \emph{weighted quasi-metric space} if $(X,d)$ is weightable by $w$ and a \emph{co-weighted quasi-metric space} if $(X,d)$ is co-weightable by $w$. 

% For any quasi-metric space $(X,d)$, we call a map $f: X\to\R$ \emph{left 1-Lipschitz} if for all $x,y \in X$, $f(x)-f(y)\leq d(x,y)$ and \emph{right 1-Lipschitz} if $f(y)-f(x)\leq d(x,y)$. It is easy to show that for any weighted quasi-metric space $(X,d,w)$, $w$ is a right 1-Lipschitz function. Similarly, for any co-weighted quasi-metric space $(X,d,w')$, $w'$ is a left 1-Lipschitz function. A function that is both left and right 1-Lipschitz is called \emph{1-Lipschitz}.

A concept strongly related to weighted quasi-metrics is that of a \emph{partial metric} \cite{Ma94}. A \emph{partial metric} on a set $X$ is a mapping $p:X\times X\to\R_+$ such that for all $x,y,z\in X$:
\begin{enumerate}[(i)]
\item $p(x,y)\geq p(x,x)$;
\item $x=y\iff p(x,x)=p(y,y)=p(x,y)$;
\item $p(x,y)=p(y,x)$;
\item $p(x,z)\leq p(x,y)+p(y,z)-p(y,y)$.
\end{enumerate}
It has been shown \cite{Ma94} that there is a bijection between the partial metrics and generalized weighted quasi-metrics: the transformation $d(x,y) =p(x,y)-p(x,x)$ produces a generalized weighted quasi-metric with weight function $x\mapsto p(x,x)$ out of a partial metric while the $p(x,y)=q(x,y)+w(x)$ produces a partial metric out of a generalized weighted quasi-metric.

\section{Edit distance}\label{sec:editdist}

Waterman, Smith and Beyer, in their 1976 paper \cite{WSB76}, introduced a general form of the edit distance on sets of words, henceforth referred to as the WSB distance. It was constructed by defining a set of allowed weighted transformations between two strings and then minimizing the sum of weights of allowed operations transforming (in the sense of ordered composition) one word into another. They also proposed an algorithm to compute the WSB distance based on dynamic programming.

In this section, we present a recursive definition of edit distance on a free semigroup that generalizes that of Waterman, Smith and Beyer and describe some of its most important properties. The edit distance provides the conceptual and algorithmic foundation to both global and local similarities on free semigroups. Before producing the main definition, we formalize the concept of a \emph{gap penalty}, which we will discuss in detail later in the text.

\begin{defn}
Let $\Sigma$ be a set. A positive function $\gamma:\Sigma^+\to\R$ is called a \emph{gap penalty} over $\Sigma^+$ if for all $u,v\in\Sigma^+$,
\begin{equation}
\gamma(u) + \gamma(v) \geq \gamma(uv).
\end{equation}
We denote by $\Gamma(\Sigma)$ the set of all gap penalties over $\Sigma^+$.
\end{defn}

\begin{defn}
Let $\Sigma$ be a set, $d:\Sigma\times\Sigma\to\R$, and $\alpha$ and $\beta$ be functions $\Sigma^+\to\R$ such that $\alpha^p,\beta^p\in\Gamma(\Sigma)$. Let $x,y\in\Sigma^*$ and let $m=\abs{x}$ and $n=\abs{y}$. Let $1\leq p<\infty$ and define the distance $\sD:\Sigma^*\times\Sigma^*\to\R$ using the following recursion:
\begin{enumerate}[(a)]
\item $\sD(\bar{x}_0,\bar{y}_0)=\sD(e,e)=0$,
\item $\sD(e,\bar{y}_j)= \alpha(\bar{y}_j)$ for all $1\leq j\leq n$,
\item $\sD(\bar{x}_i,e)=\beta(\bar{x}_i)$ for all $1\leq i\leq m$, and
\item for all $1\leq i\leq m$ and $1\leq j\leq n$
\begin{equation*}
\begin{split}
\sD(\bar{x}_i,\bar{y}_j) =\Bigg(  \min\Bigg\lbrace & \sD^p(\bar{x}_{i-1},\bar{y}_{j-1})+d^p(x_i,y_j), \\
& \min_{1\leq k\leq j}\left\{\sD^p(\bar{x}_{i},\bar{y}_{j-k})+\alpha^p(y_{j-k+1}\ldots y_j)\right\},\\ 
& \min_{1\leq k\leq i}\left\{\sD^p(\bar{x}_{i-k},\bar{y}_{j})+ \beta^p(x_{i-k+1}\ldots x_i)\right\} \Bigg\rbrace \Bigg) ^{1/p}.
\end{split}
\end{equation*}
\end{enumerate}
The $\ell^p$ \emph{edit distance} between the sequences $x$ and $y$ (extending $d$, $\alpha$ and $\beta$), is then given by $\sD(x,y)=\sD(\bar{x}_m,\bar{y}_n)$.
\end{defn}

\begin{rem}\label{rem:gappenalties}
We have assumed that $\alpha^p,\beta^p\in\Gamma(\Sigma)$ instead of just being positive functions in order to have $\sD(e,x)=\alpha(x)$ and $\sD(x,e)=\beta(x)$ for all $x\in\Sigma^+$. For a general positive function $\alpha:\Sigma^+\to\R$, the function $\gamma$, given recursively for all $x\in\Sigma^+$ by $\gamma(x_1)=\alpha^p(x_1)$ and
\begin{equation}
\gamma(\bar{x}_i)=  \min_{1\leq k\leq i}\left\{\gamma(\bar{x}_{i-k})+\alpha^p(x_{i-k+1}\ldots x_i)\right\},
\end{equation}
will belong to $\Gamma(\Sigma)$ and therefore $\gamma^{1/p}$ can be used in definition of $\sD$ instead of $\alpha$.
\end{rem}

\begin{rem}
Also note that the distance $\sD$ as defined does not extend $d$ from $\Sigma$ in the strict sense, that is, it is not necessarily true that for all $a,b\in\Sigma$, $\sD(a,b)=d(a,b)$. However, this statement does become correct if we additionally assume $d^p(a,b)\leq \beta^p(a)+\alpha^p(b)$. 
\end{rem}

\begin{rem}\label{rem:WSBalgo}
The $\ell^p$ edit distance between $x$ and $y$ can be computed using dynamic programming algorithm of Waterman, Smith and Beyer \cite{WSB76}. Let $\mD$ be an $(m+1)\times(n+1)$ matrix with rows and columns indexed from $0$ such that $\mD_{0,0}=0$ and for all $i=1,2\ldots m$ and $j=1,2\ldots n$, $\mD_{i,0}= \beta(\bar{x}_i)$, $\mD_{0,j}= \alpha(\bar{y}_j)$,
and  
\begin{equation}\label{eqn:WSBdefn3}
\begin{split}
\mD_{i,j} =\min\Bigg\lbrace & \mD_{i-1,j-1}+d^p(x_i,y_j),\\
& \min_{1\leq k\leq j}\left\{\mD_{i,j-k}+\alpha^p(y_{j-k+1}\ldots y_j)\right\},\\ 
& \min_{1\leq k\leq i}\left\{\mD_{i-k,j}+\beta^p(x_{i-k+1}\ldots x_i)\right\} \Bigg\rbrace.\\
\end{split}
\end{equation}
Then, we have $\sD(x,y)=(\mD_{m,n})^{1/p}$. The original WSB distance is obtained when $p=1$.
\end{rem}

\subsection{Alignments}

From the recursive definition, it follows that the $\ell^p$ edit distance $\sD(x,y)$ can be decomposed as the $\ell^p$ sum of the distances of non-overlapping factors of $x$ and $y$. This decomposition provides an optimal \emph{alignment} between $x$ and $y$.

\begin{defn}[\cite{SW81a}]
Let $x,y\in\Sigma^*$. An \emph{alignment} between $x$ and $y$ is a finite sequence of pairs $\algn{(x^*_k,y^*_k)}_{k=1}^K$, where $x=x^*_1x^*_2\ldots x^*_K$, $y=y^*_1y^*_2\ldots y^*_K$ and for each $1\leq k\leq K$ either
\begin{enumerate}[(a)]
\item $x^*_k=x_i$ and $y^*_k=y_j$ for some $i,j$, or
\item $x^*_k\in\mathfrak{F}(x)$, $x^*_k\neq e$  and $y^*_k=e$, or
\item $x^*_k=e$, $y^*_k\in\mathfrak{F}(y)$ and $y^*_k\neq e$.
\end{enumerate}
We will use $\mathcal{A}(x,y)$ to denote the set of all alignments of $x$ and $y$.
\end{defn}

Each pair $(x^*_k,y^*_k)$ corresponds to an \emph{edit operation} that transforms $x^*_k$ into $y^*_k$. Pairs of the form $(a,b)$, $(x,e)$ and $(e,y)$ where $a,b\in\Sigma$ and $x,y\in\Sigma^+$ represent a \emph{substitution} of the letter $a$ for the letter $b$, \emph{deletion} of the word $x$ and \emph{insertion} of the word $y$, respectively. Insertions and deletions are collectively called \emph{indels}.

Every transformation $(x^*_k,y^*_k)$ can be given a weight or a cost equal to $\sD(x^*_k,y^*_k)$, with the weight of an alignment $\algn{(x^*_k,y^*_k)}_{k=1}^K$ being equal to the $\ell^p$ sum of the weights of the individual transformations. The distance $d$ on $\Sigma$ provides \emph{substitution costs}, while the values of $\alpha$ and $\beta$, give the costs of indels. Thus, the edit distance between $x$ and $y$ can be described as the minimum weighted cost (in the $\ell^p$ sense) of transforming the sequence $x$ into $y$ using substitutions and indels as edit operations. This provides an alternative characterization of edit distance, which was long known for the $\ell^1$ case \cite{SW81a} and which we state here in general form without proof as Lemma \ref{lemma:aligndecomp} below. 

\begin{lem}\label{lemma:aligndecomp}
Let $\Sigma$ be a set, $d:\Sigma\times\Sigma\to\R$, and $\alpha,\beta:\Sigma^+\to\R_+$. Suppose $\sD$ is an $\ell^p$ edit distance on $\Sigma^*$ with respect to $d$, $\alpha$ and $\beta$. Then, for all $x,y\in\Sigma^*$
\begin{equation}\label{eq:decomp}
\sD(x,y) = \min \biggl\{ \left( {\textstyle \sum_{k=1}^K \sD^p(x^*_k,y^*_k)} \right)^{1/p}\  \big\vert\ \algn{(x^*_k,y^*_k)}_{k=1}^K\in\mathcal{A}(x,y) \biggr\}.
\end{equation}
\qed
\end{lem}

\subsection{Edit distances as quasi-metrics}

We now proceed to state the conditions for an $\ell^p$ edit distance to be a quasi-metric. For simplicity we restrict ourselves to edit distances with gap penalties that are increasing and depend solely on fragment composition and length, while more general gap penalties are considered in Appendix \ref{app:editdist}.

\begin{defn}\label{defn:incrgap}
Let $\Sigma$ be a set. We call a function $\gamma:\Sigma^*\to\R$ \emph{increasing} if 
for all $u,v,x\in\Sigma^*$,
\begin{equation}
\gamma(uxv)\geq \gamma(uv).
\end{equation}
\end{defn}

\begin{defn}
Let $\Sigma$ be a set. A function $\gamma\in\Gamma(\Sigma)$ is called a \emph{composition-length gap penalty} on $\Sigma^+$ if it is increasing and has a form
\begin{equation}
\gamma(z) = \sum_i \phi(z_i) + \psi(\abs{z})
\end{equation}
for all $z\in\Sigma^+$, where $\phi$ is a map $\Sigma\to\R$ and $\psi$ is a function $\N\to\R$. We denote by $\gmcl$ the set of all composition-length gap penalties on $\Sigma^+$.
\end{defn}

Composition-length gap penalties have a component solely dependent on the length of the inserted or deleted word and a composition-dependent component. Current applications of edit distances in computational biology (see for example \cite{Gusfield97}) mainly use gap penalties that are the same for insertions and deletions and depend solely on the fragment length, thus satisfying our definition of composition-length gap penalties with $\phi=0$. We chose the above definition in order to include all such cases and to provide simple but sufficiently general gap penalties for consideration of global and local similarities. The requirement for composition-length gap penalties to be increasing is included because it is a necessary condition for applications of our main Theorem \ref{thm:localqm}.

The most widely used length-dependent gap penalty functions are \emph{linear}, of the form $\psi(k)=\mu k$, and \emph{affine}, of the form $\psi(k)=\mu+\nu k$, where $\mu,\nu$ are constants. The main advantage of affine gap penalties is that the dynamic programming algorithm for computation of distances in this case can be modified to run in $O(nm)$ average and worst case time, where $m=\abs{x}$ and $n=\abs{y}$ \cite{Go82}, as opposed to $O(m^2n+mn^2)$  for the most general WSB algorithm \cite{WSB76}. Gap penalties of the form $\psi(k)=\mu+\nu\log(k)$ have also been considered \cite{Waterman84b}. Note that the algorithmic complexity of the WSB algorithm for distances using composition-length gap penalties depends mainly on the form of $\psi$ since the composition-dependent component is linear.

\setcounter{thmeditdistqmn}{\value{thm}} 
\setcounter{thmeditdistqms}{\value{section}} 

\begin{thm}\label{thm:editdistqm}
Let $\Sigma$ be a set and let $1\leq p<\infty$. Suppose $d$ is a separating quasi-metric on $\Sigma$ and $\gamma,\delta\in\gmcl$ such that for all $a,b\in\Sigma$,
\begin{equation}
\gamma(b)-\gamma(a) \leq d^p(a,b)
\end{equation}
and
\begin{equation}
\delta(a)-\delta(b) \leq d^p(a,b).
\end{equation}
Let $\alpha=\gamma^{1/p}$ and $\beta=\delta^{1/p}$. Then, the $\ell^p$ edit distance $\sD$, extending $d$,$\alpha$ and $\beta$, is a separating quasi-metric on $\Sigma^*$. \qed
\end{thm}

Theorem \ref{thm:editdistqm} is a generalization of similar theorems for $p=1$ proven by Waterman \emph{et al.} \cite{WSB76} for constant substitution costs and gap penalties depending on fragment length, and by Spiro and Macura \cite{SM04} in a more general setting. We state and prove a version with fewer restriction on gap penalties as Theorem \ref{thm:editdistqm1} in Appendix \ref{app:editdist}.

\begin{rem}\label{rem:invqmsemigroup}
According to \cite{RoSch02}, a quasi-metric $d$ defined on a semigroup $(X,\star)$ is called \emph{invariant} with respect to $\star$ if for all $x,y,z\in X$,
\begin{equation}
d(x\star z,y\star z)\leq d(x,y) \quad \text{and}\quad d(z\star x,z\star y)\leq d(x,y).
\end{equation}
It is apparent from the definition that the edit distance $\sD$ on the free semigroup $\Sigma^*$, which satisfies Theorem \ref{thm:editdistqm}, is invariant with respect to the string concatenation.
\end{rem}

Since our $\ell^p$ edit distances depend on several parameters, we introduce a nomenclature to make this explicit.
\begin{defn}
Let $\Sigma$ be a set and let $1\leq p<\infty$. Suppose $\sD$ is an $\ell^p$ edit distance extending a quasi-metric $d$ on $\Sigma$ and gap penalties $\alpha,\beta$ such that $\alpha^p,\beta^p\in\gmcl$. We will write $\sD=\DQ{p}{d}{\alpha}{\beta}$ if $\sD$ is a quasi-metric and $\sD=\DM{p}{d}{\alpha}$ if $\sD$ is a metric (it is necessary that $\alpha=\beta$ if $\sD$ is a metric).
\end{defn}

Most (if not all) instances of edit distances in computer science, computational biology and pure mathematics involve the $\ell^1$ edit distances. Below, we outline some of the well-known examples.

\begin{exmp}\label{ex:lev}
The \emph{Levenstein metric} \cite{Lev66} (the original `string edit distance') is the smallest number of permitted edit operations (substitutions and indels) required to transform one string into another. In our nomenclature, for a set of letters $\Sigma$, the Levenstein distance is realized as $\DM{1}{d}{\alpha}$ where $\alpha(u)=\abs{u}$ for all $u\in\Sigma^+$ and $d$ is the \emph{discrete metric}, that is, for all $a,b\in\Sigma$
\begin{equation}
d(a,b) = \begin{cases}
0 & \text{if $a=b$},\\
1 & \text{if $a\neq b$}.
\end{cases}
\end{equation}
\end{exmp}

\begin{exmp}\label{ex:sellers}
The Sellers distance, introduced by Sellers in 1974 \cite{Se74}, is a metric obtained by extension of a metric $d$ on the set $\Sigma_\dagger=\Sigma\cup\{e\}$, the set of generators plus the identity element, to the free monoid $\Sigma^*$. It is realized as $\DM{1}{d}{\alpha}$ where $\alpha(u)=\sum_i d(u_i,e)$ for all $u\in\Sigma^+$.

This construction has long been known in the theory of topological groups \cite{Pe99b} as the Graev metric \cite{Graev1948,Graev1951} on the free group $F(\Sigma)$. Recall that $F(\Sigma)$ consists of all sequences of letters from the generating set $\Sigma$ and their inverses; in other words, $F(\Sigma)=Y^*$, where $Y=\Sigma\cup \Sigma^{-1}$ and $\Sigma^{-1}$ is the set consisting of inverses of elements of $\Sigma$. Let $\rho$ be a metric on the set $Y_\dagger=Y\cup \{e\}$. The Graev metric $\bar{\rho}$ is then a maximal invariant metric on $F(\Sigma)$ such that $\bar{\rho}$ restricted to the set $Y_\dagger$ is equivalent to $\rho$. Note that the notion of invariance in this context is slightly different than the definition of an invariant quasi-metric on a semigroup from Remark \ref{rem:invqmsemigroup} above: a metric $\rho$ on a group $(X,\star)$ is called \emph{invariant} with respect to $\star$ if for all $x,y,z\in X$,
\begin{equation}
\rho(x\star z,y\star z)= \rho(z\star x,z\star y) = \rho(x,y).
\end{equation}
%It can be shown that the $\ell_1$ edit distance satisfies this stronger form of invariance as well.

The maximality of the Sellers-Graev metric can also be observed in the context of the free monoid $\Sigma^*$ using the following argument. Let $\sD=\DM{1}{d}{\alpha}$ where $d$ is a on $\Sigma$ and $\alpha$ is a gap penalty. Define a metric $d_\dagger$ on $\Sigma_\dagger$ by
\begin{equation}
d_\dagger(a,b) = \begin{cases} \sD(a,b)& \text{if $a,b\in\Sigma$},\\
\alpha(a) & \text{if $b=e$},\\
\alpha(b) & \text{if $a=e$}.\\  \end{cases}
\end{equation}
It is clear that $\sD$ extends $d_\dagger$ from $\Sigma_\dagger$ to $\Sigma^*$. However, for every $x\in\Sigma^*$, \[\sD(x,e)\leq \left(\sum_i \alpha^p(x_i)\right)^{1/p}\leq \sum_i \alpha(x_i)\] and hence every edit distance extending $d_\dagger$ to $\Sigma^*$ will be smaller than the Sellers-Graev distance.
\end{exmp}

\begin{exmp}\label{ex:lcs}
Let $\Sigma$ be a set and for $u,v\in\Sigma^*$ denote by $\mathsf{LCS}(u,v)$ the \emph{longest common subsequence} of $u$ and $v$. Define \[\rho(u,v)=\abs{u}+\abs{v}-2\abs{\mathsf{LCS}(u,v)}.\] It can be easily shown that $\rho$ is a metric on $\Sigma^*$ and that $\rho$ can be realized as $\DM{1}{d}{\alpha}$ where $\alpha(u)=\abs{u}$ for all $u\in\Sigma^+$ and  $d(a,b)=2$ for all $a,b\in\Sigma$ such that $a\neq b$ (cf. \cite{Gusfield97}, pp. 246). Since $d(a,b)\geq \alpha(a)+\alpha(b)$, the optimal alignment can be expressed solely in terms of insertions and deletions . The longest common subsequence metric provides a special case of the Sellers-Graev metric.   
\end{exmp}

\subsection{Alignment decomposition}

Recall that Lemma \ref{lemma:aligndecomp} indicates that the total $\ell^p$ edit distance $\sD$ between two words $x$ and $y$ can be optimally decomposed as an $\ell^p$ sum of the distances between constituent factors of $x$ and $y$. Lemma \ref{lemma:condP} below shows that, if the gap penalties are increasing, an arbitrary choice of a factor $y'$ of $y$ decomposes the edit distance between $x$ and $y$ into $\ell^p$ sum of the edit distances between fragments of $x$ and $y$. In this case, all of $x$ is used up while some parts of $y$ could be `lost' (Figure \ref{fig:decompose}). A similar splitting can also be achieved with a choice of a fragment of $x$. We call this property \emph{arbitrary decomposability.}

\begin{defn}
Let $\Sigma$ be a set, let $\rho:\Sigma^*\times\Sigma^*$ be a distance function on the free monoid $\Sigma^*$ and let $1\leq p<\infty$. We say that $\rho$ is \emph{arbitrarily decomposable of order $p$} if  for all $x,y\in\Sigma^*$,
\begin{enumerate}[(i)]
\item For every $y'\in\mathfrak{F}(y)$ there exist $x',x_1^*,x_2^*\in\mathfrak{F}(x)$ such that $x=x_1^*x'x_2^*$ and $y_1^*,y_2^*,u,v\in\mathfrak{F}(y)$ such that $y=y_1^*uy'v y_2^*$ and 
\begin{equation}
\rho(x,y)\geq \Big(\rho^p(x_1^*,y_1^*)+ \rho^p(x',y')+ \rho^p(x_2^*,y_2^*)\Big)^{1/p}; \tag{A1}\label{eqn:A1}
\end{equation}
\item For every $x'\in\mathfrak{F}(x)$ there exist $y',y_1^*,y_2^*\in\mathfrak{F}(y)$ such that $y=y_1^*y'y_2^*$ and $x_1^*,x_2^*,u,v\in\mathfrak{F}(x)$ such that $x=x_1^*ux'v x_2^*$ and 
\begin{equation}
\rho(x,y)\geq \Big(\rho^p(x_1^*,y_1^*)+ \rho^p(x',y')+ \rho^p(x_2^*,y_2^*)\Big)^{1/p}. \tag{A2}
\end{equation}
\end{enumerate}
\end{defn}

Note that if the distance function $\rho$ is symmetric, the two properties above collapse into a single one.

\begin{figure}[htbp]
\begin{center}
\input{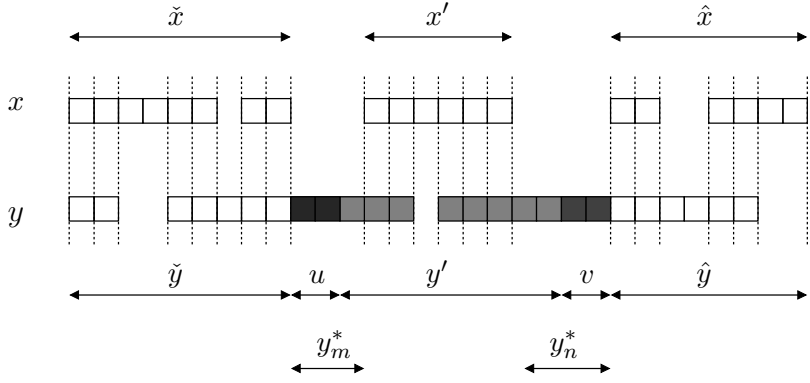}
\caption{Arbitrary decomposability (part \ref{eqn:A1}) of an alignment. A choice of $y'$ induces a decomposition of both $x$ and $y$ such that $x=\check{x}x'\hat{x}$, $y=\check{y}uy'v\hat{y}$ and $\rho(x,y)\geq \big(\rho^p(\check{x},\check{y})+ \rho^p(x',y')+ \rho^p(\hat{x},\hat{y})\big)^{1/p}$. Dashed lines indicate the boundaries of edit operations. The fragments $u$ and $v$ of $y$ are `lost': they do not contribute to decomposition.}
\label{fig:decompose}
\end{center}
\end{figure}

\begin{lem}\label{lemma:condP}
Let $\Sigma$ be a set and let $d:\Sigma\times\Sigma\to\R$. Suppose that $\alpha$ and $\beta$ are increasing functions $\Sigma^+\to\R$ such that $\alpha^p,\beta^p\in\Gamma(\Sigma)$ and $\sD$ is an $\ell^p$ edit distance on $\Sigma^*$ extending $d$, $\alpha$ and $\beta$. Then, $\sD$ is arbitrarily decomposable of order $p$.
\end{lem}

\begin{proof}
We will prove only the first part of the definition of arbitrary decomposability because the second follows by the same argument. Let $x,y\in\Sigma^*$ and let $y'\in\mathfrak{F}(y)$. By Lemma \ref{lemma:aligndecomp}, the distance $\sD(x,y)$ can be written as 
\begin{equation*}
\sD(x,y) = \left(\sum_{k=1}^K \sD^p(x^*_k,y^*_k)\right)^{1/p},
\end{equation*}
where $x=x^*_1x^*_2\ldots x^*_K$, $y=y^*_1y^*_2\ldots y^*_K$. Let $1\leq m\leq n\leq K$ be such that $y^*_m\neq e$, $y^*_n\neq e$, $y'\in\mathfrak{F}(y^*_m\ldots y^*_n)$ and $y^*_{m+1}\ldots y^*_{n-1}\in\mathfrak{F}(y')$ (i.e. $y^*_m\ldots y^*_n$ is the smallest factor of $y$ having $y'$ as a factor -- see Figure \ref{fig:decompose}). Then, the fragments $y^*_m$ and $y^*_n$ contain parts of $y'$. (Note that $y'$ always coincides with $y^*_m\ldots y^*_n$ if the gap penalties depend only on composition.)

Consider the fragment $y^*_m$. According to Lemma \ref{lemma:aligndecomp}, $y^*_m$ can be either a letter ($y^*_m\in\Sigma$) or a fragment ($y^*_m\in\Sigma^*$), since the possibility of $y^*_m=e$ was explicitly excluded.
If $y^*_m\in\Sigma$, let $u=e$ and $u'=y^*_m$ so that $\sD(x^*_m,y^*_m) = \sD(x^*_m,u')$. On the other hand, if $y^*_m\not\in\Sigma$, then by Lemma \ref{lemma:aligndecomp} $x^*_m=e$. Let $u,u'\in\Sigma^*$ be fragments of $y^*_m$ such that $y^*_m=uu'$ and $u'_1=y'_1$ (i.e. we split $y^*_m$ into a part not overlapping with $y'$ and a part overlapping with it). It is possible that $u=e$ but we always have $u'\in\Sigma^+$ by construction. By our assumption about increasing gap penalty, it follows that 
\begin{equation}
\sD(x^*_m,y^*_m) = \sD(e,uu') = \alpha(uu') \geq \alpha(u) = \sD(x^*_m,u').
\end{equation}

In a similar way, the fragment $y^*_n$ can be expressed as $y^*_n=v'v$ where $y'_{\abs{y'}}=v'_{\abs{v'}}$ (i.e. $v'$ contains the end of $y'$) and
\begin{equation}
\sD(x^*_n,y^*_n) = \sD(e,v'v) = \alpha(v'v) \geq \alpha(v) = \sD(x^*_n,v').
\end{equation}

Now, let $\check{x}=x^*_1\ldots x^*_{m-1}$, $x'=x^*_m\ldots x^*_{n}$ and $\hat{x}=x^*_{n+1}\ldots x^*_K$. Let $\check{y}=y^*_1\ldots y^*_{m-1}$ and $\hat{y}=y^*_{n+1}\ldots y^*_K$. Then, $x=\check{x}x'\hat{x}$, $y=\check{y}uy'v\hat{y}$ and
\begin{align*}
\sD(x,y) & = \left(\sum_{k=1}^K \sD^p(x^*_k,y^*_k)\right)^{1/p}\\
& = \Big(\sD^p(\check{x},\check{y})+ \sD^p(x^*_m,uu') + \sum_{k=m+1}^{n-1} \sD^p(x^*_k,y^*_k) + \sD^p(x^*_n,v'v) + \sD^p(\hat{x},\hat{y})\Big)^{1/p}\\
& \geq \Big(\sD^p(\check{x},\check{y})+ \sD^p(x^*_m,u') + \sum_{k=m+1}^{n-1} \sD^p(x^*_k,y^*_k) + \sD^p(x^*_n,v') + \sD^p(\hat{x},\hat{y})\Big)^{1/p}\\
& \geq \Big(\sD^p(\check{x},\check{y})+ \sD^p(x',y')+ \sD^p(\hat{x},\hat{y})\Big)^{1/p},
\end{align*}
since $(x^*_m,u')(x^*_{m+1},y^*_{m+1})\ldots (x^*_{n-1},y^*_{n-1})(x^*_n,v')$ is an alignment of $x'$ and $y'$ and hence the $\ell^p$ sum of distances over it is greater than $\sD^p(x',y')$ by Lemma \ref{lemma:aligndecomp}.
\end{proof}

Therefore, any $\ell^p$ edit distance with composition-length gap penalties is arbitrarily decomposable of order $p$. However, there exist arbitrarily decomposable distances that are not $\ell^p$ edit distances.

\begin{exmp}\label{ex:gaplessdist}
Let $\Sigma$ be a finite set and let $d$ be a metric on $\Sigma$. For any $n\in\N$, the \emph{generalized Hamming distance} $d_n$ on $\Sigma^n$ is given for all $x,y\in\Sigma^n$ by
\begin{equation}
d_n(x,y) = \sum_{i=1}^n d(x_i,y_i).
\end{equation}
It can be easily shown that $d_n$ is a metric. The generalized Hamming distance is a natural generalization of the Hamming distance \cite{Hamming:1950} where the distance $d$ on $\Sigma$ is the discrete metric.

Let $f:\Sigma\to\R$ be a function such that for all $a,b\in\Sigma$,
\begin{equation}
\abs{f(a)-f(b)}\leq d(a,b)\leq f(a)+f(b).
\end{equation}
It immediately follows that for every $n\in\N$ and for all $x,y\in\Sigma^n$,
\begin{equation}\label{eqn:normpairHamming}
\abs{\bar{f}(x)-\bar{f}(y)}\leq d_n(x,y)\leq \bar{f}(x)+ \bar{f}(y).
\end{equation}
Define the distance $\rho:\Sigma^*\times\Sigma^*\to\R$ by extending $d_n$ and $f$ so that for all $x,y\in\Sigma^*$,
\begin{equation}
\rho(x,y) = \begin{cases} d_n(x,y) & \text{if $\abs{x}=\abs{y}=n$,}\\ \bar{f}(x)+ \bar{f}(y) & \text{if $\abs{x}\neq\abs{y}$}. \end{cases}
\end{equation}

Using (\ref{eqn:normpairHamming}), it is easy to show that $\rho$ is a metric on $\Sigma^*$. Furthermore, $\rho$ is arbitrarily decomposable (of order $1$). Indeed, consider $x,y\in\Sigma^*$ and $y'\in\mathfrak{F}(y)$. If $\abs{x}=\abs{y}$, one immediately obtains the required decomposition using the form of the generalized Hamming distance. On the other hand, if $\abs{x}\neq\abs{y}$, we have 
\begin{equation}
\rho(x,y) = \bar{f}(x)+ \bar{f}(y) \geq \rho(x,e) + \rho(e,y')
\end{equation}
leading to the decomposition where $x'=e$ and $u$ and $v$ take all of $y$ apart from $y'$.

The metric $\rho$ (generalized to $\ell^p$ form) can be interpreted as an `ungapped' version of edit distances. Here substitutions are allowed only between sequences of equal length and the function $\bar{f}$ plays a role of gap penalty so that the only way to transform sequences of unequal length is through a full deletion followed by insertion.
\end{exmp}

\section{Global Similarity}\label{sec:globalsim}

A more common approach to sequence comparison is to maximize similarities instead of minimizing distances. In this case a similarity measure on $\Sigma$ and gap penalties are used to define the similarity between two sequences in $\Sigma^*$ using the Needleman-Wunsch \cite{NW70} or Smith-Waterman \cite{SW81} dynamic programming algorithm,  which are very similar to the algorithm for computation of edit distances described above. As in the case of $\ell^p$ edit distances above, we define sequence similarities using a recursive definition.

\begin{defn}
Let $\Sigma$ be a set, $s:\Sigma\times\Sigma\to\R$, and let $\gamma,\delta\in\Gamma(\Sigma)$. For any $x,y\in\Sigma^*$ where $m=\abs{x}$ and $n=\abs{y}$, define the \emph{global (Needleman-Wunsch) similarity} $\sS:(x,y)\mapsto\R$ using the following recursion:
\begin{enumerate}[(a)]
\item $\sS(\bar{x}_0,\bar{y}_0)=\sS(e,e)=0$,
\item $\sS(e,\bar{y}_j)=-\gamma(\bar{y}_j)=$ for all $1\leq j\leq n$,
\item $\sS(\bar{x}_i,e)=-\delta(\bar{x}_i) $ for all $1\leq i\leq m$, and 
\item for all $1\leq i\leq m$ and $1\leq j\leq n$
\begin{equation}
\begin{split}
\sS(\bar{x}_i,\bar{y}_j) =\max\Bigg\lbrace & \sS(\bar{x}_{i-1},\bar{y}_{j-1})+s(x_i,y_j), \\
& \max_{1\leq k\leq j}\left\{\sS(\bar{x}_{i},\bar{y}_{j-k})-\gamma(y_{j-k+1}\ldots y_j)\right\},\\ 
& \max_{1\leq k\leq i}\left\{\sS(\bar{x}_{i-k},\bar{y}_{j})-\delta(x_{i-k+1}\ldots x_i)\right\} \Bigg\rbrace .
\end{split}
\end{equation}
\end{enumerate}
The global similarity between the sequences $x$ and $y$ (extending $s$, $\gamma$ and $\delta$), is defined by $\sS(x,y)=\sS(\bar{x}_m,\bar{y}_n)$.
\end{defn}

The algorithm used to compute $\ell^1$ edit distance (Remark \ref{rem:WSBalgo}) can also be used for computation of similarities by setting $d=-s$, $\alpha=\gamma$ and $\beta=\delta$, computing $\sD$ for $p=1$ and then taking $\sS=-\sD$. The running time of the dynamic programming algorithm depends on the properties of gap penalties, as discussed in the previous section. Note that the gap penalty functions are positive in the case of both distances and similarities, being added in the former case and subtracted in the latter. It is also possible to express global similarity as a sum of similarities over alignments, as is done for edit distance in Lemma \ref{lemma:aligndecomp}.

\begin{exmp}\label{ex:lcs1}
It is well known \cite{Gusfield97} that the longest common subsequence problem described in Example \ref{ex:lcs} can be approached using similarities rather than distances. Let $\Sigma$ be a set and let $s$ be a scoring function on $\Sigma$ such that $s(a,b)=0$ if $a\neq b$ and $s(a,a)=1$. Let $\gamma(x)=\delta(x)=0$ for all $x\in\Sigma^+$. It is easy to confirm that for $x,y\in\Sigma^*$, $\sS(x,y)=\abs{\mathsf{LCS}(x,y)}$. 
\end{exmp}

Relations between global similarities and $\ell^1$ edit distances were explored early on \cite{SWF81,SW81a}. 

\begin{thm}[\cite{SWF81,SW81a}]
Let $\sS$ be the global similarity with respect to $s,\gamma$ and $\delta$ such that 
for all $x\in\Sigma^+$, $\gamma(x)=\delta(x)=\psi(\abs{x})$, where $\psi$ is a positive function. Consider the $\ell^1$ edit distance $\sD$, extending $d:\Sigma\times\Sigma\to\R$ and the gap penalties $\alpha$ and $\beta$ and let $s_M=\max\{s(a,b)\ |\ a,b\in\Sigma\}$. Suppose for all $a,b\in\Sigma$
\begin{equation}
d(a,b) = s_M - s(a,b),
\end{equation}
and for all $x\in\Sigma^+$,
\begin{equation}
\alpha(x)=\beta(x)=\frac{s_M\abs{x}}{2} + \psi(\abs{x}).
\end{equation}
Then, $\sS$ and $\sD$ will induce equivalent sets of optimal alignments and for all $x,y\in\Sigma^*$,
\begin{equation}
\sD(x,y) = s_M\frac{\abs{x}+\abs{y}}{2} - \sS(x,y). 
\end{equation}
\qed
\end{thm}

The distance function obtained by taking a constant minus similarity is not guaranteed to satisfy any of the axioms for a metric or a quasi-metric: one problem is that the self-similarity $\sS(x,x)$ for any $x\in\Sigma^*$ is not necessarily a constant. However, under some more restrictive but frequently valid assumptions, it is possible to transform similarities into metrics or quasi-metrics. We establish the results that have interesting biological interpretations and provide the foundation for considering transformation of local similarities, discussed in Section \ref{sec:localsim}, to quasi-metrics. 

\begin{defn}
Let $X$ be a set and let $s$ be a (similarity) map $X\times X\to\R$. We call $s$ a \emph{sane} scoring function if for all $x,y\in X$,
\begin{enumerate}[(i)]
\item $s(x,x) > 0$, 
\item $s(x,x) \geq s(x,y)$, and
\item $s(x,x) \geq s(y,x)$.
\end{enumerate}
\end{defn}

Thus, a similarity map is sane if every element of $\Sigma$ `keeps its identity' with respect to it. Every point is similar to itself and this similarity cannot be smaller than similarity to any other point.

\setcounter{lemselfsimn}{\value{thm}} 
\setcounter{lemselfsims}{\value{section}} 
\begin{prop}\label{lemma:selfsim}
Let $\Sigma$ be a set and let $s:\Sigma\times\Sigma\to\R$ be a a sane scoring function over $\Sigma$. Suppose $\gamma,\delta\in\Gamma(\Sigma)$ and $\sS$ the global similarity on $\Sigma^*$ with respect to $s,\delta$ and $\gamma$. Then, $\sS$ is a sane scoring function and for all $x\in\Sigma^*$,  
\begin{equation}
\sS(x,x) = \sum_{i=1}^{\abs{x}} s(x_i,x_i).
\end{equation}
\qed
\end{prop}

Proposition \ref{lemma:selfsim} and Theorem \ref{thm:editdistqm} give us a straightforward way to convert global similarities to (quasi-) metrics. Since this transformation is based on the transformations of similarity scores to distances on generators, we first introduce additional nomenclature.

\begin{defn}\label{defn:alphnomen}
Let $\Sigma$ be a set and let $1\leq p<\infty$. For a sane scoring function $s$ on $\Sigma$, we will use $\AQ{p}{s}$ to denote the distance $q$ on $\Sigma$ given by
\begin{equation}\label{eqn:aqdef}
q(a,b)=\bigl(s(a,a)-s(a,b)\bigr)^{1/p}
\end{equation}
and 
$\AM{p}{s}$ to denote the distance $d$ on $\Sigma$ given by
\begin{equation}\label{eqn:amdef}
d(a,b)=\bigl(s(a,a)+s(b,b)-s(a,b)-s(b,a)\bigr)^{1/p}.
\end{equation}
\end{defn}
Note that at this stage we do not make an assumption that $\AQ{p}{s}$ is a quasi-metric nor that $\AM{p}{s}$ is a metric.

\setcounter{corglobdistn}{\value{thm}} 
\setcounter{corglobdists}{\value{section}} 
\begin{cor}\label{cor:simdist2}
Let $\Sigma$ be a set and let $1\leq p<\infty$. Suppose $s$ is a sane scoring function on $\Sigma$, $d=\AQ{p}{s}$ is a quasi-metric on $\Sigma$ and $\gamma,\delta\in\gmcl$ such that
\begin{equation}\label{eq:gapcond1}
\gamma(b)-\gamma(a) \leq d^p(a,b)
\end{equation}
and
\begin{equation}\label{eq:gapcond2}
s(a,a)+\delta(a)-s(b,b)-\delta(b) \leq d^p(a,b).
\end{equation}
Let $\sS$ be the global similarity with respect to $s,\gamma$ and $\delta$ and let $\alpha(x)=\gamma(x)^{1/p}$ and $\beta(x)=\bigl(\sS(x,x)+\delta(x)\bigr)^{1/p}$ for all $x\in\Sigma^+$. Then, the $\ell^p$ edit distance $\sD=\DQ{p}{d}{\alpha}{\beta}$ is given for all $x,y\in\Sigma^*$ by the formula 
\begin{equation}\label{eqn:qm}
\sD(x,y) = \Bigl(\sS(x,x) - \sS(x,y)\Bigr)^{1/p}.
\end{equation}
\qed
\end{cor}

As with edit distances, we now introduce a nomenclature for quasi-metrics and metrics obtained from similarities.
\begin{defn}
Let $\Sigma$ be a set and let $1\leq p<\infty$. Suppose $\sD$ is an $\ell^p$ edit distance obtained from a global similarity $\sS$ on $\Sigma^*$ using the formula (\ref{eqn:qm}) of Corollary \ref{cor:simdist2}, where $\sS$ extends $s:\Sigma\times\Sigma$ and $\gamma,\delta\in\gmcl$. We will write $\sD=\GQ{p}{s}{\gamma}{\delta}$ if $\sD$ is a quasi-metric and $\sD=\GM{p}{s}{\gamma}{\delta}$ if $\sD$ is a metric.
\end{defn}

The above nomenclature is redundant, in that every distance derived from similarities using Corollary \ref{cor:simdist2} can be expressed using the nomenclatures for edit distances and distances on $\Sigma$ introduced in Definition \ref{defn:alphnomen}. We have chosen to nevertheless introduce the additional notation in order to emphasize that the distances on the free monoid are derived from similarities and also because the computation of distances can be performed using algorithms for similarities. This notation will also be convenient in the following sections, where local similarities are discussed.

%The $\GM{p}{s}{\gamma}{\delta}$ metric may still have different gap penalties used to construct the similarity score as illustrated in the following example.

\begin{exmp}\label{ex:GMp}
Let $\Sigma$ be a set and suppose $s$ is a sane symmetric function $\Sigma\times\Sigma\to\R$ and $\gamma\in\gmcl$, depending only on length. This is a very frequent setup in pairwise comparison of DNA and protein sequences (see Section \ref{sec:matrices} below for more detailed discussion). Define for all $a,b\in\Sigma$, $s'(a,b)=2s(a,b)-s(b,b)$ and for all $x\in\Sigma^+$, $\gamma'(x)=2\gamma(x)+\sum_i s(x_i,x_i)$ and $\delta'(x)=2\gamma(x)$.

Suppose that the distance $d=\AQ{p}{s'}=\AM{p}{s}$ is a metric on $\Sigma$. Since $s$ is sane, $s'$ is also sane and we have \[\abs{s'(a,a)-s'(b,b)}= \abs{s(a,a)-s(b,b)}\leq d^p(a,b).\] Therefore, since $\gamma$ depends 
solely on length, the requirements (\ref{eq:gapcond1}) and (\ref{eq:gapcond2}) of Corollary \ref{cor:simdist2} are satisfied. Let $\sS$ be the global similarity extending $s,\gamma$ and $\gamma$ and let $\sS'$ be the global similarity extending $s',\gamma'$ and $\delta'$. We conclude that the distance $\sD$ given by
\begin{equation}
\sD(x,y) = \bigl(\sS'(x,x) - \sS'(x,y)\bigr)^{1/p} = \bigl(\sS(x,x) + \sS(y,y) - 2\sS(x,y)\bigr)^{1/p}
\end{equation}
is the metric $\GM{p}{s'}{\gamma'}{\delta'}$. This metric can also be expressed as $\DM{p}{\AM{p}{s}}{\alpha}$, where $\alpha(x)=\bigl(S(x,x)+\gamma(x)\bigr)^{1/p}$ for all $x\in\Sigma^+$.
\end{exmp}

\section{Local Similarity}\label{sec:localsim}

Local similarity is computed using the Smith-Waterman algorithm \cite{SW81}.

\begin{defn}
Let $\Sigma$ be a set, $s:\Sigma\times\Sigma\to\R$, and let $\gamma,\delta\in\Gamma(\Sigma)$. Let $x,y\in\Sigma^*$, $m=\abs{x}$ and $n=\abs{y}$. The \emph{Smith-Waterman} dynamic programming matrix, denoted $\SW(x,y,s,\gamma,\delta)$, is an $(m+1)\times(n+1)$ matrix $\mH$ with rows and columns indexed from $0$ such that $\mH_{0,0}=0$ and for all $1\leq i\leq m$ and $1\leq j\leq n$, $\mH_{i,0}=0$, $\mH_{0,j}=0$ and 
\begin{equation*}
\begin{split}
\mH_{i,j}=\max\bigg\lbrace & \mH_{i-1,j-1}+s(x_i,y_j), \max_{1\leq k\leq i}\left\{\mH_{i-k,j}-\delta(x_{i-k+1}\ldots x_i)\right\},\\ &  \max_{1\leq k\leq j}\left\{\mH_{i,j-k}-\gamma(y_{j-k+1}\ldots y_j)\right\},\quad 0 \bigg \rbrace.
\end{split}
\end{equation*}
The \emph{local similarity} between the sequences $x$ and $y$ (given $s$, $\gamma$, and $\delta$), denoted $\sH(x,y)$, is defined to be the largest entry of $\mH$, that is, $\sH(x,y)= \max_{i,j} \mH_{i,j}$.
\end{defn}

Local similarity between two words can be realized as global similarity of their fragments. 

\begin{thm}[\cite{SW81a}]\label{thm:globloc2}
Let $\Sigma$ be a set, $s:\Sigma\times\Sigma\to\R$ and $\gamma,\delta\in\Gamma(\Sigma)$. Suppose $\sS$ is a global similarity extending $s,\gamma$ and $\delta$ and $\sH$ is the local similarity with respect to $s,\gamma$ and $\delta$. Then, for all $x,y\in\Sigma^*$,
\begin{equation}
\sH(x,y)=\max_{\substack{x'\in\mathfrak{F}(x) \\ y'\in\mathfrak{F}(y)}}\sS(x',y').
\end{equation}
\qed
\end{thm}

Although conversion of global similarities to distances outlined in Section \ref{sec:globalsim} is relatively straightforward, its counterpart for local similarity is much less so. We now use the results from the previous sections to state our main result: construction of quasi-metrics which include conversions of local similarities.

\begin{thm}\label{thm:localqm}
Let $\Sigma$ be a set and let $1\leq p<\infty$. Let $\rho$ be a separating quasi-metric on $\Sigma^*$ that is arbitrarily decomposable of order $p$. Suppose $f$ is a strictly positive and $g$ is a non-negative function $\Sigma\to\R$ and $\bar{f}$ and $\bar{g}$ are the canonical homomorphic extensions of $f$ and $g$, respectively, to the free monoid $\Sigma^*$. Assume also that for all $x,y\in\Sigma^*$,
\begin{equation}\label{eqn:holderrho}
\bar{f}(x)-\bar{f}(y)\leq \rho^p(x,y)\quad \text{and}\quad \bar{g}(y)-\bar{g}(x)\leq \rho^p(x,y).
\end{equation}
Then, the function $\sQ:\Sigma^*\times\Sigma^* \to\R$ defined by 
\begin{equation}\label{eq:localqm}
\sQ(x,y) = \min_{\substack{\tilde{x}\in\mathfrak{F}(x) \\ \tilde{y}\in\mathfrak{F}(y)}} \left\{ \Bigl( \bar{f}(x) - \bar{f}(\tilde{x}) + \bar{g}(y) - \bar{g}(\tilde{y}) +\rho^p(\tilde{x},\tilde{y}) \Bigr)^{1/p} \right\}
\end{equation}
is a quasi-metric on $\Sigma^*$.
\end{thm}

\begin{proof}
Let $x,y,z\in\Sigma^*$. Since $\bar{f}(x)\geq\bar{f}(\tilde{x})$ and $\bar{g}(y)\geq\bar{g}(\tilde{y})$ for any $\tilde{x}\in\mathfrak{F}(x)$, $\tilde{y}\in\mathfrak{F}(y)$ and since $\rho$ is a quasi-metric and hence positive, it follows that $\sQ(x,y)\geq 0$. Furthermore, it is clear that $\sQ(x,x)=0$.

Suppose that $\sQ(x,y)=0$. Then, there exist $\tilde{x}\in\mathfrak{F}(x)$ and $\tilde{y}\in\mathfrak{F}(y)$ such that $\bar{f}(x) - \bar{f}(\tilde{x})+ \bar{g}(y) - \bar{g}(\tilde{y}) +\rho^p(\tilde{x},\tilde{y})=0$. Since $\rho(\tilde{x},\tilde{y})\geq 0$, $\bar{f}(x) - \bar{f}(\tilde{x})\geq 0$ and $\bar{g}(y) - \bar{g}(\tilde{y})\geq 0$ for any $\tilde{x},\tilde{y}\in\Sigma^*$, it follows that $\bar{f}(x) = \bar{f}(\tilde{x})$, $\bar{g}(y) = \bar{g}(\tilde{y})$ and $\rho(\tilde{x},\tilde{y})=0$. The first statement implies that $x=\tilde{x}$ since $f$ is a strictly positive function, while the last means that $\tilde{x}=\tilde{y}$ (since $\rho$ is a separating quasi-metric). Therefore, $\sQ(x,y)=0$ implies $x\in\mathfrak{F}(y)$. Hence, $\sQ(x,y)=\sQ(y,x)=0$ implies $x\in\mathfrak{F}(y)$ and $y\in\mathfrak{F}(x)$ and thus $x=y$.

To establish the triangle inequality suppose that
\begin{equation}
\sQ(x,y) = \Bigl( \bar{f}(x) - \bar{f}(\tilde{x})+ \bar{g}(y) - \bar{g}(\tilde{y})+ \rho^p(\tilde{x},\tilde{y}) \Bigr)^{1/p}
\end{equation}
for some $\tilde{x}\in\mathfrak{F}(x)$, $\tilde{y},\in\mathfrak{F}(y)$ and 
\begin{equation}
\sQ(y,z) = \Bigl( \bar{f}(y) - \bar{f}(\dot{y})+ \bar{g}(z) - \bar{g}(\dot{z})+ \rho^p(\dot{y},\dot{z}) \Bigr)^{1/p}
\end{equation}
for some $\dot{y}\in\mathfrak{F}(y)$ and $\dot{z}\in\mathfrak{F}(z)$. Write out $\tilde{y}=y_iy_{i+1}\ldots y_{i+m-1}$, $\dot{y}=y_jy_{j+1}\ldots y_{j+n-1}$ where $m=\abs{\tilde{y}}$, $n=\abs{\dot{y}}$, $1 \leq i\leq i+m-1\leq \abs{y}$ and $1 \leq j\leq j+n-1\leq \abs{y}$. If $\tilde{y}$ and $\dot{y}$ overlap, that is, if $i \leq j\leq m$ or $j \leq i\leq n$, let $y'$ denote the whole overlapping fragment (for example, if $i \leq j\leq i+m-1\leq i+n-1$, $y'=y_jy_{j+1}\ldots y_{i+m-1}$ -- see Figure \ref{fig:overlap}). If $\tilde{y}$ and $\dot{y}$ do not overlap or either $\tilde{y}$ or $\dot{y}$ is identity, let $y'=e$. 

\begin{figure}[htbp]
\begin{center}
\input{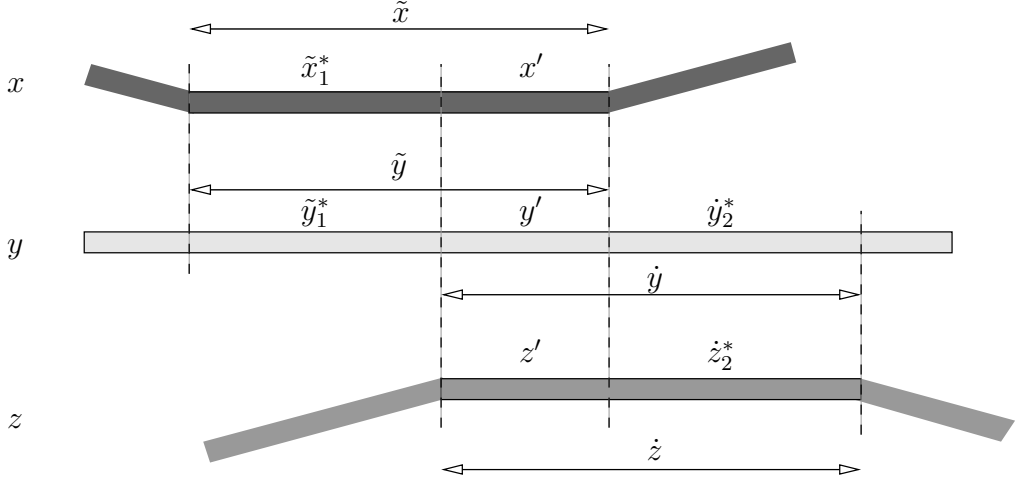}
\caption{Decomposition of $x$, $y$ and $z$. In this pattern of overlap of $\tilde{y}$ and $\dot{y}$, we have $\tilde{x}^*_2=\tilde{y}^*_2=\dot{y}^*_1=\dot{z}^*_1=e$.}
\label{fig:overlap}
\end{center}
\end{figure}

Since $\rho$ is arbitrarily decomposable of order $p$, there exist $x',\tilde{x}_1^*,\tilde{x}_2^*\in\mathfrak{F}(\tilde{x})$ such that $\tilde{x}=\tilde{x}_1^*x'\tilde{x}_2^*$ and $\tilde{y}_1^*,\tilde{y}_2^*,u,v\in\mathfrak{F}(\tilde{y})$ such that $\tilde{y}=\tilde{y}_1^*uy'v \tilde{y}_2^*$ and 
\begin{equation}
\rho(\tilde{x},\tilde{y})\geq \Big(\rho^p(\tilde{x}_1^*,\tilde{y}_1^*)+ \rho^p(x',y')+ \rho^p(\tilde{x}_2^*,\tilde{y}_2^*)\Big)^{1/p}.
\end{equation}
Furthermore, by the same assumption, there exist $z',\dot{z}_1^*,\dot{z}_2^*\in\mathfrak{F}(\dot{z})$ such that $\dot{z}=\dot{z}_1^*z'\dot{z}_2^*$ and $\dot{y}_1^*,\dot{y}_2^*,\dot{u},\dot{v}\in\mathfrak{F}(\dot{y})$ such that $\dot{y}=\dot{y}_1^*\dot{u}y'\dot{v} \dot{y}_2^*$ and 
\begin{equation}
\rho(\dot{y},\dot{z})\geq \Big(\rho^p(\dot{y}_1^*,\dot{z}_1^*)+ \rho^p(y',z')+ \rho^p(\dot{y}_2^*,\dot{z}_2^*)\Big)^{1/p}.
\end{equation}

Therefore, using the Minkowski inequality, 
\begin{align*}
\sQ(x,y)+\sQ(y,z) &\geq \Bigl( \bar{f}(x) - \bar{f}(\tilde{x})+ \bar{g}(y) - \bar{g}(\tilde{y})\\
& \qquad + \rho^p(\tilde{x}_1^*,\tilde{y}_1^*)+ \rho^p(x',y')+ \rho^p(\tilde{x}_2^*,\tilde{y}_2^*) \Bigr)^{1/p}\\
& \; + \Bigl( \bar{f}(y) - \bar{f}(\dot{y}) + \bar{g}(z) - \bar{g}(\dot{z})\\ 
& \qquad + \rho^p(\dot{y}_1^*,\dot{z}_1^*)+ \rho^p(y',z')+ \rho^p(\dot{y}_2^*,\dot{z}_2^*) \Bigr)^{1/p}\\
& \geq \Bigl( \bar{f}(x) - \bar{f}(\tilde{x})+ \bar{f}(y) - \bar{f}(\dot{y}) + \rho^p(\tilde{x}_1^*,\tilde{y}_1^*) + \rho^p(\tilde{x}_2^*,\tilde{y}_2^*)\\
& \qquad + \bar{g}(y) - \bar{g}(\tilde{y})+ \bar{g}(z) - \bar{g}(\dot{z}) + \rho^p(\dot{y}_1^*,\dot{z}_1^*) + \rho^p(\dot{y}_2^*,\dot{z}_2^*) \\
& \qquad + \bigl(\rho(x',y') + \rho(y',z')  \bigr)^p \Bigr)^{1/p}.
\end{align*}

Since $\bar{f}$ and $\bar{g}$ are additive functions that satisfy the inequality (\ref{eqn:holderrho}) and since $y'$ is the full extent of the overlap between $\tilde{y}$ and $\dot{y}$, we have 
\begin{align*}
\phantom{\geq} & \quad \bar{f}(x) - \bar{f}(\tilde{x})+ \bar{f}(y) - \bar{f}(\dot{y}) + \rho^p(\tilde{x}_1^*,\tilde{y}_1^*) + \rho^p(\tilde{x}_2^*,\tilde{y}_2^*)\\
\geq &\quad  \bar{f}(x) - \bar{f}(\tilde{x}) + \bar{f}(y) - \bar{f}(\dot{y}) + \bar{f}(\tilde{x}_1^*)-\bar{f}(\tilde{y}_1^*) + \bar{f}(\tilde{x}_2^*)-\bar{f}(\tilde{y}_2^*)\\
\geq & \quad \bar{f}(x) - \bar{f}(x') \geq 0,
\end{align*}
and
\begin{align*}
\phantom{\geq} & \quad \bar{g}(y) - \bar{g}(\tilde{y})+ \bar{g}(z) - \bar{g}(\dot{z}) + \rho^p(\dot{y}_1^*,\dot{z}_1^*) + \rho^p(\dot{y}_2^*,\dot{z}_2^*)\\
\geq &\quad  \bar{g}(y) - \bar{g}(\tilde{y})+ \bar{g}(z) - \bar{g}(\dot{z}) - \bar{g}(\dot{y}_1^*) +\bar{g}(\dot{z}_1^*) -\bar{g}(\dot{y}_2^*) +\bar{g}(\dot{z}_2^*)\\
\geq & \quad \bar{g}(z) - \bar{g}(z') \geq 0.
\end{align*}
Hence, by the triangle inequality for $\rho$, \[\sQ(x,y)+\sQ(y,z) \geq \Bigl( \bar{f}(x) - \bar{f}(x')+ \bar{g}(z) - \bar{g}(z') +\rho^p(x',z') \Bigr)^{1/p}
 \geq \sQ(x,z),\] as required. 
\end{proof}

\begin{rem}
We have shown in the separation part of the proof of Theorem \ref{thm:localqm}  above that $\sQ(x,y)=0\implies x\in\mathfrak{F}(y)$ and hence the associated partial order of the quasi-metric $Q$ is $x\leq_\sQ y \iff x\in\mathfrak{F}(y)$. If $g$ is a strictly positive function, $\sQ$ is a separating quasi-metric and the partial order is trivial: $\sQ(x,y) = 0$ implies $x=y$ and hence each point is only comparable to itself.

However, if $g$ is zero everywhere, then $\sQ(x,y) = 0$ and $x\neq y$ implies that $x$ is a factor of $y$ while $y$ is not a factor of $x$, so that $x$ and $y$ are non-trivially comparable. In this case, the quasi-metric $\sQ$ is not separating and it generalizes the substring partial order: for every $x,y\in\Sigma^*$ such that $x\in\mathfrak{F}(y)$, we have $\sQ(x,y)=0$. Therefore, $\sQ(x,y)$ can be interpreted as measuring how far is $x$ from being a factor of $y$.

Since the identity $e$ is a trivial factor of every word and $\bar{f}(e)=0$, it follows that (in the case of $g\equiv 0$) $\sQ(e,x)=0$ for every $x\in\Sigma^+$, in contrast to $\rho(e,x)\geq 0$. On the other hand,
it can be easily seen that $\sQ(x,e)=(\bar{f}(x))^{1/p}$ and hence for all $y\in\Sigma^+$, $\sQ(x,y)\leq (\bar{f}(x))^{1/p} = \sQ(x,e)$. 
\end{rem}

We now introduce a nomenclature for quasi-metrics and their associated metrics defined in Theorem \ref{thm:localqm}.
\begin{defn}
Let $\Sigma$ be a set and let $1\leq p<\infty$. Suppose $\rho$ is a separating quasi-metric on $\Sigma^*$ and $f$ and $g$ are functions $\Sigma\to\R$ that satisfy all the requirements of Theorem \ref{thm:localqm} with respect to $\rho$ and $p$. Let $\sQ$ be the quasi-metric obtained using the formula (\ref{eq:localqm}) of Theorem \ref{thm:localqm}. We will write $\sQ=\LQ{p}{\rho}{f}{g}$ if $\sQ$ is a quasi-metric and $\sQ=\LM{p}{\rho}{f}{g}$ if $\sQ$ is a metric.
\end{defn}

\begin{rem}
Edit distances described in Section \ref{sec:editdist} are always global: they measure the full cost of transformation between two words in $\Sigma^*$. Indeed, a truly `local' distance, that is the distance measured on factors of words being compared, would not satisfy the triangle inequality.

The $\mathsf{LQ}^p$ distances are slightly different. The distance $\rho$ contributes to $\sQ$ by evaluating the pair of factors $\tilde{x}$ and $\tilde{y}$ that are `closest' to each other (relative to $\bar{f}$ and $\bar{g}$), while $\bar{f}$ and $\bar{g}$ score the left-over pieces of $x$ and $y$, respectively. The extent of $\tilde{x}$ and $\tilde{y}$ relative to $x$ and $y$ depends on the exact choice of functions $f$ and $g$ and their relation to the distance $\rho$. For example, when $f$ and $g$ are very large compared to $\rho$, the factors $\tilde{x}$ and $\tilde{y}$ will approach the whole sequences $x$ and $y$. On the other hand, if $f$ and $g$ are small, they will contribute most to $\mathsf{LQ}^p(\rho,f,g)$, depending on the exact properties of $\rho$.

When both $f$ and $g$ are strictly positive, the $\mathsf{LQ}^p$ distance has a global character in that the whole of $x$ and $y$ are accounted for. If $g\equiv 0$, only $x$ contributes to the distance as a whole; the sequence $y$ contributes only through its factor closest to a factor of $x$. In general, it is possible to favor $x$ or $y$ by appropriately choosing the values of $f$ and $g$.
\end{rem}

Theorem \ref{thm:localqm} can be applied to similarities in the following manner. Let $\sQ=\LQ{p}{\rho}{f}{g}$. Define a global similarity $\sigma$ on $\Sigma^*$ by
\begin{equation}
\sigma(x,y) = \bar{f}(x) + \bar{g}(y) - \rho^p(x,y).
\end{equation}
Then,
\begin{align}
\sQ(x,y) & = \Bigl(\bar{f}(x) + \bar{g}(y) - \max_{\substack{\tilde{x}\in\mathfrak{F}(x) \\ \tilde{y}\in\mathfrak{F}(y)}} \bigl\{ \bar{f}(\tilde{x}) + \bar{g}(\tilde{y}) - \rho^p(\tilde{x}, \tilde{y}) \bigr\} \Bigr)^{1/p} \notag \\
& = \Bigl(\bar{f}(x) + \bar{g}(y) - \max_{\substack{\tilde{x}\in\mathfrak{F}(x) \\ \tilde{y}\in\mathfrak{F}(y)}} \sigma(\tilde{x}, \tilde{y}) \Bigr)^{1/p} \label{eq:localsimqm}.
\end{align}
Hence, if $\sigma$ can be computed using the Needleman-Wunsch algorithm (that is, if $\rho$ is an $\ell^p$ edit distance), then $\sQ$ can always be evaluated by using the Smith-Waterman algorithm to compute the local similarity $H(x,y)=\max\{\sigma(\tilde{x}, \tilde{y})\ |\  \tilde{x}\in\mathfrak{F}(x), \tilde{y}\in\mathfrak{F}(y)\}$ and then using Equation (\ref{eq:localsimqm}).

Since the functions $f$ and $g$ as well as the quasi-metric $\rho$ are arbitrary, the applicability of Theorem \ref{thm:localqm} to similarities is very wide. The following examples are simple corollaries of Theorem \ref{thm:localqm} and the results in Sections \ref{sec:editdist} and \ref{sec:globalsim} that have important uses in computational biology.

\begin{exmp}
Let $\Sigma$ be a finite set and suppose $s$ is a sane symmetric function $\Sigma\times\Sigma\to\R$ such that the distance $d=\AQ{1}{s}$, is a metric on $\Sigma$. Let $\mu=\min\{s(a,b)\ |\ a,b\in\Sigma\}$ and let $f(a)=s(a,a)-\mu$. It is clear from the definitions of $f$ and $d$ that $\abs{f(a)-f(b)}\leq d(a,b)\leq f(a)+f(b)$. 

Let $\rho$ be the arbitrarily decomposable metric extending the generalized Hamming distance based on $d$ and $f$ to $\Sigma^*$, as in Example \ref{ex:gaplessdist} and define $g:\Sigma\to\R$ by $\alpha(a)=s(a,a)$. By Theorem \ref{thm:localqm} we can construct the distance $\LQ{1}{\rho}{g}{g}$, which is in fact the metric $\LM{1}{\rho}{g}{g}$. The underlying similarity $\sigma$, given by Equation (\ref{eq:localsimqm}), is 
\begin{equation}
\sigma(x,y) = \begin{cases} 2s_n(x,y) & \text{if $\abs{x}=\abs{y}=n$,}\\ 2\mu & \text{if $\abs{x}\neq\abs{y}$}. \end{cases}
\end{equation}
where $s_n(x,y)=\sum_{i=1}^n s(x_i,y_i)$. In computational biology applications, $\mu$ will be negative (there will be at least two points in $\Sigma$ that are dissimilar) and hence the local similarity will always be realized by aligning the fragments of the same length. Therefore, the local similarity based on $\sigma$ is \emph{gapless similarity}, which has considerable historical importance since the first version of BLAST \cite{AGMML90} suite of tools for sequence database search based on local similarities used a heuristic that computed gapless alignments. Gapless alignments had an advantage that they could be computed faster and the statistics of similarity scores arising from them were well characterized \cite{KaA90,KA93}.
\end{exmp}

In the following examples \ref{ex:qmlocsim1}, \ref{ex:qmlocsim2} and \ref{ex:mlocsim}, we will assume that $s:\Sigma\times\Sigma\to\R$ is a sane scoring function, $\gamma,\delta\in\gmcl$ only depend on length and $S$ and $H$ are global and local similarity with respect to $s$, $\gamma$ and $\delta$, respectively. In addition, let $f(a)=s(a,a)$ for all $a\in\Sigma$.

\begin{exmp}
\label{ex:qmlocsim1}
Suppose $\AQ{1}{s}$ is a quasi-metric. By Corollary \ref{cor:simdist2}, the distance $D$ on $\Sigma^*$ given by $D(x,y)=S(x,x)-S(x,y)$ is a quasi-metric $\GQ{1}{s}{\gamma}{\delta}$. Consider the distance $Q=\LQ{1}{\GQ{1}{s}{\gamma}{\delta}}{f}{0}$. It is easy to see that $\bar{f}(x)=S(x,x)=H(x,x)$ and hence
\begin{equation}\label{eqn:sim2qm01}
Q(x,y) = S(x,x) - \max_{\substack{\tilde{x}\in\mathfrak{F}(x) \\ \tilde{y}\in\mathfrak{F}(y)}} S(\tilde{x}, \tilde{y}) = H(x,x) - H(x,y).
\end{equation}
As remarked earlier, the partial order associated with $Q$ in this case is subfragment partial order. Furthermore, the triangle inequality for $Q$ is equivalent to 
\begin{equation}\label{eqn:sim2qm011}
H(x,y) + H(y,z) \leq H(y,y) + H(x,z).
\end{equation}

If $H$ is symmetric, that is, if $s$ is symmetric and $\gamma=\delta$, we have 
\begin{equation}
Q(x,y) + H(y,y) = Q(y,x) + H(x,x),
\end{equation}
and hence $Q$ is a co-weightable quasi-metric and $-H$ is a partial metric. Note that in this case, the triangle inequality (\ref{eqn:sim2qm011}) is exactly equivalent to the triangle inequality for the symmetrization $M(x,y)=Q(x,y)+Q(y,x)$ (Example \ref{ex:mlocsim}), that is, if $M$ is a metric then $Q$ is a quasi-metric.

The fact that Equation (\ref{eqn:sim2qm01}) gives a quasi-metric was first established in \cite{Stojmirovic2004}. Indeed, the two generate equivalent neighborhoods: for any $x\in\Sigma^*$, the set of all points $y\in\Sigma^*$ such that $H(x,y)>\kappa$ is equal to the set $\{y\in\Sigma^*:Q(x,y)<\e\}$ where $\e=S(x,x)-\kappa$. 
\end{exmp}

\begin{exmp}
\label{ex:qmlocsim2}
Recall the notation from Example \ref{ex:GMp}, where $s$ is symmetric, $\gamma=\delta$, $s'(a,b)=2s(a,b)-s(b,b)$, $\gamma'(x)=2\gamma(x)+\sum_i s(x_i,x_i)$ and $\delta'(x)=2\gamma(x)$. Let $S'$ and $H'$ be global and local similarity with respect to $s'$, $\gamma'$ and $\delta'$, respectively.

Suppose that $\AQ{p}{s'}$ is a quasi-metric (equivalently that $\AM{p}{s}$ is a metric) and consider the quasi-metric $Q'=\LQ{p}{\GM{p}{s'}{\gamma'}{\delta'}}{f}{0}$. By the argument of Example \ref{ex:qmlocsim1}, 
\begin{equation}
Q'(x,y)=\bigl(H'(x,x)-H'(x,y)\bigr)^{1/p} =\bigl(S(x,x)-H'(x,y)\bigr)^{1/p}.
\end{equation}
However, in this case the local similarity
\begin{equation}
H'(x,y)= \max_{\tilde{x}, \tilde{y}} S'(\tilde{x}, \tilde{y}) = \max_{\tilde{x}, \tilde{y}} (2S(x,y) - S(y,y))
\end{equation}
is clearly asymmetric. This similarity score has, to our knowledge, never been previously used for sequence comparison, although it can be easily computed using Smith-Waterman algorithm (provided that the particular implementation used allows composition-length gap penalties). It has the advantage that it is still true that $H'$ is topologically equivalent to $Q'$ and that $Q'$ corresponds to the subfragment partial order. 

The asymmetry of $H'$ may be exploited to favor the integrity of one sequence over the other in biological sequence alignments. For example, in cases where translated DNA sequences are compared to proteins, it is desirable to emphasize the protein sequence, which is `real' (experimentally established), at the expense of translated DNA sequences, which is only hypothetical. We intend to evaluate the broad utility of using variants of $H'$ and $Q'$ for biological sequence comparisons in a subsequent publication.
\end{exmp}

\begin{exmp}
\label{ex:mlocsim}
Making the same assumptions as in Example \ref{ex:qmlocsim2} above, consider the metric $M=\LM{p}{\GM{p}{s'}{\gamma'}{\delta'}}{f}{f}$. It is easy to see that $M$ is indeed a metric given by
\begin{equation}\label{eqn:sim2m01}
M(x,y)=\bigl(H(x,x)+H(y,y)-2H(x,y)\bigr)^{1/p}.
\end{equation}

Equation (\ref{eqn:sim2m01}), for $p=1$, was extensively considered in computer science and computational biology. The $\mathsf{LCS}$ similarities (Examples \ref{ex:lcs} and \ref{ex:lcs1}) are related to distances in this way. Linial \emph{et al.} \cite{LLTY97} proposed using $M$ as a distance on sets of protein sequences but did not explicitly prove it was a metric. Spiro and Macura \cite{SM04} have given the conditions under which $M$ is indeed a metric. Since $H$ is here assumed symmetric, this result is equivalent to $\LQ{1}{\GQ{1}{s}{\gamma}{\delta}}{f}{0}$ being a quasi-metric (Example \ref{ex:qmlocsim1}), established by Stojmirovi\'c \cite{Stojmirovic2004} under slightly different assumptions. Itoh \emph{et al.} \cite{IGAK05} derived the same result as a corollary of a more general inequality for similarities that relied on the finiteness of the generator alphabet. In a poster abstract \cite{Fischer02}, Fischer proposed the general form of Equation (\ref{eqn:sim2m01}) with arbitrary $p$ as a way to convert similarities to distances and stated without proof the conditions for $M$ to be a metric.

For $p=2$, the form of Equation (\ref{eqn:sim2m01}) resembles the formula for the canonical metric in inner-product vector spaces. In this case, $x$ and $y$ would be vectors and $H$ would be a positive-definite bilinear form.
\end{exmp}

Theorem \ref{thm:localqm} can be applied in the context of free abelian monoids with no change. We illustrate this by a very simple, followed by a more biologically relevant example.

\begin{exmp}
Let $\Sigma$ be the set of all prime numbers and let $\Sigma^*$ be the free abelian monoid over $\Sigma$ under multiplication (i.e. the set of natural numbers $\N$). Let $d_\dagger$ be a discrete metric on $\Sigma$ (here we implicitly assume that $\Sigma$ includes $1$) and let $f(a)=1$ and $\alpha(a)=0$ for all $a\in\Sigma$. Let $\rho$ be the Sellers-Graev metric extension of $d_\dagger$ to $\N$. It is clear that $\rho(x,y)$ is just the number of different prime factors between $x$ and $y$ (the non-matching prime factors are matched to $1$) and that it is arbitrarily decomposable. Hence, we can apply Theorem \ref{thm:localqm} to obtain a quasi-metric $\sQ$, so that $\sQ(x,y)$ is the number of prime factors of $x$ not in common to $y$. The global similarity $\sigma$ on $\N$ (here equivalent to local similarity), given by  $\sigma(x,y)=\bar{f}(x)-\rho(x,y)$ evaluates to the number of common prime factors (excluding $1$) between $x$ and $y$.
\end{exmp}

\begin{exmp}
Let $\Sigma$ be a finite set and let $A(\Sigma^k)$ denote the free abelian monoid generated by the set of all words of length exactly $k$ (we will call $z\in\Sigma^k$ a $k$-tuple). Members of $A(\Sigma^k)$ are therefore multisets of $k$-tuples. Now consider the same structure as in the previous example.

Let $d_\dagger$ be a discrete metric on $\Sigma^k\cup\{e\}$ and let $f(a)=1$ and $\alpha(a)=0$ for all $a\in\Sigma$. Let $\rho$ be the Sellers-Graev metric extension of $d_\dagger$ to $A(\Sigma^k)$. All requirements of Theorem \ref{thm:localqm} still apply. The value $\sQ(x,y)$ is the number of $k$-tuples that are contained in $x$ but not in $y$ and the global (and local) similarity $\sigma$ gives the number of $k$-tuples common to both $x$ and $y$.

The similarity $\sigma$ has been used in computational biology as a computationally inexpensive approximation of global similarity between two sequences \cite{KMKM02,Edgar04}. Each sequence is mapped to $A(\Sigma^k)$ by taking the multiset of all of its (overlapping) $k$-tuples and the similarity $\sigma$ is used to approximate the global similarity $S$.
\end{exmp}

\section{Scoring Functions on Generators}\label{sec:matrices}

In the previous sections we have made no assumption on the set of generators $\Sigma$ and all our results apply to arbitrary sets. However, as we noted before, the principal objects 
motivating our results are sets of biological sequences and profiles derived from them. The former two sets are finite and therefore the scoring functions over them are given by \emph{score matrices}. We therefore proceed to discuss the similarity and distance measures on the sets of nucleotides, amino acids and profiles and their applicability to our theory.

\subsection{Nucleotide scoring matrices}

The nucleotide alphabet consists of only 4 letters (\texttt{A}, \texttt{C}, \texttt{G}, and \texttt{T}) and the score matrices most frequently used for database search depend on only two parameters, for scoring a match or a mismatch of two nucleotides. For example, the \textsl{blastn} program, a part of the BLAST \cite{altschul97gapped} suite of tools for sequence database search based on local similarities, which searches a DNA database with a DNA sequence as a query, uses the scoring matrix of the form 
\begin{equation}
s(a,b)=\begin{cases}
        5 & \text{if $a=b$}\\
        -4 & \text{if $a\neq b$}.\\
       \end{cases}
\end{equation}
The above scoring function is obviously sane and the distance $d=\AQ{p}{s}$ is a discrete metric for any $1\leq p<\infty$. Therefore, all match/mismatch scoring schemes satisfy the requirements of Theorem \ref{thm:editdistqm} and its corollaries. 

More complex score matrices, where transitions (changes \texttt{C}$\leftrightarrow$\texttt{T} and \texttt{A}$\leftrightarrow$\texttt{G}) have different scores than transversions (all other mutations) have been proposed for improving the accuracy of database searches \cite{SGA91,Durbin:1998}. It is easy to show that the distance $\AQ{1}{s}$ (and hence  $\AQ{p}{s}$ for all $p$) will still satisfy the triangle inequality and hence be a metric if the value of distance associated by transition is not greater than twice the transversion distance. Since the likelihood and hence the similarity score of transition is larger than that of transversion, this condition is very likely to be satisfied in practice. For example, all scoring matrices examined by States \emph{et al.} \cite{SGA91} satisfy this condition and are sane.

\subsection{Amino acid scoring matrices}\label{subsec:BLOSUM}

Unlike the nucleotide alphabet, the standard amino acid alphabet consists of 20 amino acids of markedly different chemical properties and structural roles. Hence, the regularly used 
amino acid scoring matrices are much more complex than the matrices over nucleotides discussed above. Many amino acid scoring matrices were developed over the years for various purposes, including sequence similarity search, structural prediction and phylogenetic analysis \cite{NKK88,TK96,KOK99}. Most of them arise from analysis of sets of peptide sequences known to be to a certain extent related. 

Dayhoff \emph{et al.} \cite{Dayhoff:1978} proposed in 1970s the family of scoring matrices called PAM, which were based on a Markov model of evolution of proteins. PAM matrices were the original standard choice for sequence comparison. Several improved versions of PAM matrices were constructed later \cite{Gonnet:1992,JTT92,MSV02,MV00,VST03}, in order to address some of the deficiencies arising from lack of sufficient data at the time of the construction of the original PAM family. For PAM-like matrices, the larger the number appended to their name (such as PAM-$n$), the sequences to be compared are assumed to have more diverged in evolution.

Presently, the most widely used family of scoring matrices is BLOSUM, derived by Henikoff and Henikoff in 1992 \cite{Henikoff:1992} using an empirical procedure. In particular, the BLOSUM62 matrix has long been believed to be among the best performing matrices for general sequence similarity search \cite{HH93} and is used as default by BLAST (more specifically, the \textsl{blastp} program). In contrast to the PAM-like matrices, the larger the number appended to the name of a BLOSUM matrix, the more the sequences to be compared are assumed to be closely related.

In addition to the above mentioned families, some score matrices were constructed specifically for searches involving transmembrane regions of proteins \cite{JTT94,MRR01,NHH00} while others were derived from structural alignments in order to improve sensitivity of searches involving distantly related proteins \cite{PDS00,KQG00,BC01}.
\begin{table}[!hbt]
\begin{center}
{\small
\begin{tabular}{lcrrr}
Matrix & Reference & $\mathsf{AQ}^1$ & $\mathsf{AQ}^2$ & $\mathsf{AM}^2$ \\ \hline
PAM40 & \cite{Dayhoff:1978} & 28 & 0 & 0 \\
PAM120 & \cite{Dayhoff:1978} & 88 & 0 & 0 \\
PAM250 & \cite{Dayhoff:1978} & 168 & 21 &  0 \\
GONNET & \cite{Gonnet:1992} & 144 & 0 &  0 \\
BLOSUM45 & \cite{Henikoff:1992}& 0 & 0 & 0 \\
BLOSUM50 & \cite{Henikoff:1992}& 0 & 0 & 0 \\
BLOSUM62 & \cite{Henikoff:1992}& 0 & 0 & 0 \\
BLOSUM80 & \cite{Henikoff:1992}& 0 & 0 & 0 \\
JTT & \cite{JTT92}& 170 & 34 & 34 \\
JTTtm & \cite{JTT94}& 214 & 18 & 20 \\
BC0030 & \cite{BC01}& 214 & 12 & 4 \\
SDM & \cite{PDS00}& 134 & 0 & 0 \\
HSDM & \cite{PDS00}& 142 & 6 & 0 \\
OPTIMA & \cite{KQG00}& 74 & 15 & 2 \\
PHAT75/73 & \cite{NHH00}& 6 & 0 & 0 \\
VTML160 & \cite{MSV02}& 28 & 0 & 0 \\
VTML250 & \cite{MSV02}& 100 & 14 & 0 \\
dist.20comp & \cite{CB05}& 0 & 0 & 0 \\
PMB120 & \cite{VST03}& 0 & 0 & 0 \\
PMB250 & \cite{VST03}& 8 & 3 & 0 \\ \hline
\end{tabular}
}
\end{center}
\caption{Number of triples of amino acids failing the triangle inequality for distances derived from various symmetric score matrices. All the matrices are considered over the standard (20 letter) amino acid alphabet (that is, excluding non-standard letters representing more than one amino acid). Due to symmetry of similarity scores, the triangle inequalities for $\mathsf{AQ}^1$ and $\mathsf{AM}^1$ are equivalent and the column for $\mathsf{AM}^1$ is omitted. 
} \label{tbl:BLOSUMqm} 
\end{table}

Table \ref{tbl:BLOSUMqm} shows the numbers of violations of the triangle inequality for the distances $\mathsf{AQ}^1$, $\mathsf{AQ}^2$ and $\mathsf{AM}^2$ obtained from several common (symmetric) score matrices. The matrices featured in Table \ref{tbl:BLOSUMqm} are all sane and represent only a very small sample of all existing amino acid score matrices that are most frequently used and cited. 

All of the scoring matrices mentioned so far were symmetric with the exception of the SLIM family \cite{MRR01} for comparison of transmembrane proteins. Yu \emph{et al.} \cite{YWA03} recently proposed a concept of \emph{compositionally adjusted score matrices}, which are asymmetric and which can be derived from symmetric score matrices by considering different background frequencies of amino acids in the first vs. the second sequence. The rationale for compositional adjustment is that some proteins, especially from organisms with biased amino acid usage, can have significantly different background frequencies of amino acids, than the ones used to construct the standard matrices. It was demonstrated in \cite{YWA03} that using compositional adjustment results in improvement of sensitivity of pairwise sequence comparison.

\begin{table}[!hbt]
\begin{center}
{\small
\begin{tabular}{l|rrrr|rrrr}
 & \multicolumn{4}{c}{\textit{C. tetani}} & \multicolumn{4}{c}{\textit{M. tuberculosis}} \\
Matrix & $\mathsf{AQ}^1$ & $\mathsf{AQ}^2$ & $\mathsf{AM}^1$ & $\mathsf{AM}^2$ & $\mathsf{AQ}^1$ & $\mathsf{AQ}^2$ & $\mathsf{AM}^1$ & $\mathsf{AM}^2$ \\ \hline
PAM40  & 36 & 0 & 36 & 0  & 40 & 0 & 40 & 0  \\
PAM120  & 129 & 0 & 126 & 0  & 113 & 0 & 116 & 0  \\
GONNET  & 152 & 0 & 152 & 0  & 151 & 0 & 150 & 0  \\
BLOSUM45  & 0 & 0 & 0 & 0  & 4 & 0 & 4 & 0  \\
BLOSUM50  & 1 & 0 & 2 & 0  & 3 & 0 & 2 & 0  \\
BLOSUM62  & 1 & 0 & 2 & 0  & 1 & 0 & 2 & 0  \\
BLOSUM80  & 0 & 0 & 0 & 0  & 0 & 0 & 0 & 0  \\
JTT  & 353 & 11 & 378 & 0  & 320 & 5 & 330 & 0  \\
BC0030  & 234 & 3 & 244 & 4  & 249 & 2 & 272 & 4  \\
SDM  & 132 & 0 & 132 & 0  & 132 & 8 & 132 & 0  \\
HSDM  & 144 & 1 & 144 & 0  & 143 & 0 & 142 & 0  \\
OPTIMA  & 77 & 4 & 78 & 2  & 78 & 2 & 80 & 2  \\
PHAT75/73  & 10 & 0 & 12 & 0  & 19 & 0 & 26 & 0  \\
VTML160  & 32 & 0 & 34 & 0  & 42 & 0 & 50 & 0  \\
dist.20comp  & 0 & 0 & 0 & 0  & 0 & 0 & 0 & 0  \\
PMB120  & 0 & 0 & 0 & 0  & 0 & 0 & 0 & 0  \\
\hline
\end{tabular}
}
\end{center}
\caption{Number of triples of amino acids failing the triangle inequality for various compositionally adjusted asymmetric score matrices. Each matrix was adjusted from a symmetric matrix by using the composition of either \textit{C. tetani} or \textit{M. tuberculosis} proteome as the first set of frequencies, together with the implicit amino acid frequencies from BLOSUM62 as the second set of frequencies.
} \label{tbl:compadj} 
\end{table}

Table \ref{tbl:compadj} shows the violations of the triangle inequality for the distances obtained from some of the matrices from Table \ref{tbl:BLOSUMqm}, adjusted to take into account the amino acid compositions of proteomes of bacterial species \textit{Clostridium tetani} and \textit{Mycobacterium tuberculosis}. Both of these species have compositionally biased genomes and proteomes. The matrices were constructed using a Newtonian procedure described in \cite{YA05} and \cite{AWGAMSY05}. The background distribution for the second sequence comes from the original BLOSUM62 matrix. In this way, the constructed similarity scores and distances can be used to compare sequences known to come from the above organisms to sequences from general datasets.

Table \ref{tbl:BLOSUMqm} and Table \ref{tbl:compadj} demonstrate that most scoring matrices, both symmetric and asymmetric, can be converted to the $\mathsf{AM}^2$ metric while many can be converted to $\mathsf{AQ}^2$ quasi-metric as well. In contrast, most matrices fail the triangle inequalities for $\mathsf{AQ}^1$ and $\mathsf{AM}^1$. Therefore, our generalization of edit distances and related sequence similarities to $\ell^p$ form allows us to use a much wider class of matrices to construct (quasi-) metrics on the set of all protein sequences. This is in contrast to the $\ell^1$-type results from the previous work \cite{Stojmirovic2004,SM04}, which only apply to the BLOSUM family plus a few more similar matrices.

\subsection{Profiles}\label{subsec:profiles}

Recall (Example \ref{exmp:profiles}) that given a set $\Sigma$, a profile over $\Sigma$ is a word in the free monoid $\mathcal{M}(\Sigma)^*$, that is, a finite sequence of finite measures over $\Sigma$. In biological applications, $\Sigma$ is finite and therefore a profile $x$ can be treated as a sequence of vectors $\vec{x}_i\in\R^n$, where $n=\abs{\Sigma}$. For each $i$, the vector $\vec{y}=\vec{x}_i$ has non-negative entries. In some applications, it is further assumed that $\vec{y}$ is a probability distribution, that is, that $\sum_j y_j = 1$.

In biological context, profiles represent generalized sequences over the basic alphabet $\Sigma$ where each position has a probability distribution of letters instead of a single letter. They were originally introduced by Gribskov \emph{et al.} \cite{Gribskov:1987} in order to improve sensitivity of homology search by considering the information contained in multiple alignments of related proteins to query sequence databases. To do so, a \emph{Position Specific Score Matrix} or PSSM, which gives a similarity score for each letter in $\Sigma$ for each position in the query profile, is constructed. The profile-sequence comparison using PSSM can then be performed using the dynamic programming algorithms such as Needleman-Wunsch or Smith-Waterman. Profiles can also be used directly in probabilistic Hidden Markov Models \cite{Eddy98}. Profile-based homology searches are widely used and have been shown in general to be more sensitive than sequence database searches with normal sequences as queries \cite{altschul97gapped,Eddy98}.

Profiles can also be compared to other profiles as members of the free monoid $\mathcal{M}(\Sigma)^*$ using distances or similarities discussed in Sections \ref{sec:editdist}, \ref{sec:globalsim} and \ref{sec:localsim}: all that is necessary is to assign a distance or similarity measure on $\mathcal{M}(\Sigma)$ and gap penalties. Many scoring schemes were proposed in due course and we present only a few examples below. For a more detailed overview we refer the reader to the papers of Edgar and Sj\"{o}lander \cite{ES04} and Marti-Renom \emph{et al.} \cite{MMS04}, which study their performance for aligning distantly related protein sequences.

Let $\vec{x},\vec{y}\in\R^n$ be two measures in $\mathcal{M}(\Sigma)$ and let $\prf{s}$ and $\hat{d}$ denote a similarity and a distance function, respectively. The symbol $\norm{\cdot}$ denotes the $\ell^2$ norm on $\R^n$.

\begin{exmp}
The simplest similarity score between two vectors, used in CLUSTALW software for multiple sequence alignment \cite{THG94} (see also Section \ref{sec:future}) is to compute their average over a score matrix $s$ on $\Sigma$: 
\begin{equation}
\prf{s}(\vec{x},\vec{y}) = \sum_i \sum_j x_iy_j s(x_i,y_j).
\end{equation}
In general, $\AQ{p}{\prf{s}}$ and $\AM{p}{\prf{s}}$ are not a quasi-metric or a metric, respectively.
\end{exmp}

\begin{exmp}
A natural candidate for similarity score between two vectors $\vec{x}$ and $\vec{y}$ is their dot product, used in \cite{RJLG00}:
\begin{equation}
\prf{s}(\vec{x},\vec{y}) = \vec{x}\cdot \vec{y} = \sum_j x_jy_j.
\end{equation}
Clearly, $\prf{d}=\AM{2}{\prf{s}}$ is the standard Euclidean distance:
\begin{equation}
\prf{d}(\vec{x},\vec{y}) = \norm{\vec{x}-\vec{y}} = \sqrt{\sum_j (x_j-y_j)^2}.
\end{equation}
\end{exmp}

\begin{exmp}
A variation of the above is the correlation coefficient or cosine of the angle between two vectors used in the LAMA algorithm \cite{Pietrokovski96}:
\begin{equation}
\prf{s}(\vec{x},\vec{y}) = \frac{\vec{x}\cdot \vec{y}}{\norm{x}\norm{y}} = \frac{\sum_j x_jy_j}{\sqrt{\sum_j x^2 \sum_j y^2_j}}.
\end{equation}
Here $\prf{d}=\AM{2}{\prf{s}}$ can be easily shown to satisfy the triangle inequality. In general, $\prf{d}$ does not separate points, but if $\vec{x}$ and $\vec{y}$ are assumed to be probability vectors, then $\prf{d}$ is indeed a metric.
\end{exmp}

\begin{exmp}
The Jensen-Shannon divergence between two probability vectors $\vec{x}$ and $\vec{y}$, denoted $D^{\textrm{JS}}$ is given by
\begin{equation}
D^{\textrm{JS}}(\vec{x},\vec{y}) = \frac12\sum_i\left[x_i\log\frac{2x_i}{x_i+y_i} + y_i\log\frac{2y_i}{x_i+y_i} \right].
\end{equation}
While $D^{\textrm{JS}}$ is not a metric, taking the square root, that is, letting
$\prf{d}(\vec{x},\vec{y}) = \sqrt{D^{\textrm{JS}}(\vec{x},\vec{y})}$ does give a metric \cite{ES03}. Yu \cite{Yu07} proposed using this metric to compare probability distributions that are components of profiles while Yona and Levitt \cite{YL02} used the following similarity score:
\begin{equation}
\prf{s}(\vec{x},\vec{y}) = \biggl(1-D^{\textrm{JS}}(\vec{x},\vec{y})\biggr)\left(1+D^{\textrm{JS}}\left(\frac{\vec{x}+\vec{y}}{2},\bpi\right)\right),
\end{equation}
where $\bpi$ denotes a background distribution.
\end{exmp}

The above examples suggest that $\ell^2$-type edit distances and global and local similarities arising from them, could be appropriate for profile-profile comparisons.

\section{Applications and Future Directions}\label{sec:future}

Our results provide a way to construct a large variety of metrics and quasi-metrics on free semigroups. In particular, we are able to extend the conversion of similarity score matrices into alphabet (generator) distances, to the corresponding conversions of sequence similarities, global and local, to sequence distances. Hence, we are able to treat biosequence sets as spaces with geometry. The metric and quasi-metric structures provide a much richer framework than the topologies induced from them: for biosequences, $\Sigma$ is finite and hence all topologies induced from $\ell^p$ edit distances or local similarity (quasi-) metrics are equivalent to the discrete topology.

In terms of statistical characterization, since we allowed more general gap penalties and asymmetric scoring matrices, the established statistics for similarities may not be fully transfered to our general distances. For this reason, to fully exploit our general formulation, it is important to further elaborate on its statistical aspects, which is beyond the scope of the current paper.

Apart from setting a general geometric framework for sequence comparison, most direct applications to biology involve clustering. For example, global clustering of protein sequences has been performed \cite{LLTY97,SLL02}, using the metric from Example \ref{ex:mlocsim} and other derivations from similarity score. However, these works did not consider quasi-metrics and partial orders that could provide a more accurate view of the global protein sequence space. Applications to indexing and multiple sequence alignment, which we discuss in more detail below, can also be considered as clustering.

\subsubsection*{Indexing for database search}\label{subsec:indexing}

One of the principal motivations for establishing the triangle inequalities for similarity scores in the literature \cite{Stojmirovic2004,SM04,IGAK05} was to accelerate similarity search of large DNA and protein sequence databases. It has been identified early on that using the full Needleman-Wunsch and Smith-Waterman dynamic programming algorithms to search sequence datasets by sequentially scanning all entries is prohibitively computationally expensive and heuristic methods such as FASTA \cite{PL88} and BLAST \cite{AGMML90,altschul97gapped} were developed. While very fast, these methods are not \emph{consistent} \cite{PS06}, that is, they are not guaranteed to retrieve all true neighbors of a given query point. Furthermore, both FASTA and BLAST sequentially scan all of the sequences in the dataset being searched. The idea behind using the triangle inequalities for accelerating similarity search is to use the intrinsic `geometry' of the dataset and the space it lies in to construct an \emph{indexing scheme} \cite{H-K-P,PS06}, a structure that allows fully retrieving a similarity query without scanning the whole dataset. A large amount of effort was spent on producing efficient indexing structures, principally concentrating on datasets that are equipped with a metric or a vector space structure: a good overview is by Hjaltason and Samet in \cite{HjSa03}.

Let $X\subset\Sigma^*$ be a finite sequence dataset. A range query of $X$ based on local similarity $H$ (depending on the score matrix $s$ and gap penalties $\gamma$ and $\delta$), centered at the query point $x\in\Sigma^*$ with threshold $\kappa$ is the set
\begin{equation}
\mathscr{Q}_H(x,\kappa)=\{y\in X: H(x,y)\geq \kappa\}.
\end{equation}
We will now consider some ways to construct an indexing structures that accelerate retrieval of $\mathscr{Q}_H$.

The first way is to consider biological sequences purely as strings with simple similarity measures often related to Levenstein distance and use string-based techniques such as hashing \cite{GiWaWaVo00,Buhler01,Kent:2002,TCOT03} or suffix arrays \cite{HuAtIr01,Hu04}. Such indexing schemes are often not consistent but may show good performance on datasets of DNA sequences where the similarity measure is very simple. For proteins, one approach was to construct a biologically meaningful metric on the amino acid alphabet and use the edit distance extension of it for sequence comparison and indexing \cite{MXSM03}. This has an advantage that existing methods for indexing metric spaces can be directly applied but ignores the need for local similarities, which cannot be converted into edit distances.

The other approach, investigated by Spiro and Macura \cite{SM04} and more thoroughly implemented by Itoh \emph{et al.} \cite{IGAK05}, was to use the inequality (\ref{eqn:sim2qm011}), which holds for some amino acid scoring matrices (Table \ref{tbl:BLOSUMqm}). The idea is to cluster proteins according to the local similarity score $H$ or associated metric $\LM{1}{\GM{1}{s'}{\gamma'}{\delta'}}{f}{f}$ (where similarity score is assumed symmetric and $f(a)=s(a,a)$ -- see Example \ref{ex:mlocsim}), and then, when searching, to compare the query sequence to centers of clusters first and only scan those clusters that overlap the query.

Note that while the neighborhoods of $\LQ{1}{\GQ{1}{s}{\gamma}{\delta}}{f}{0}$ are indeed equivalent to queries $\mathscr{Q}_H$, this is no longer true for neighborhoods of its metric symmetrization $\LM{1}{\GM{1}{s'}{\gamma'}{\delta'}}{f}{f}$. Hence, direct indexing with respect to the local similarity metric may not be optimal. Furthermore, not all similarity score matrices give rise to $\ell^1$ quasi-metrics $\AQ{1}{s}$ (Table \ref{tbl:BLOSUMqm}). Many more can be converted to $\AM{2}{s}$ and hence give rise to metrics $\LM{2}{\GM{2}{s'}{\gamma'}{\delta'}}{f}{f}$. Profile-profile comparison methods, relying on inner product for the distance between two distributions, also naturally induce $\ell^2$-type distances. None of the methods described above can efficiently cope with this situation and yet there exists a simple way to convert such similarity queries to a sequence of metric queries.

Suppose $M(x,y)=\bigl( H(x,x)+H(y,y)-2H(x,y)\bigr)^{1/p}$ is a metric for some symmetric local similarity $H$. Let 
\begin{equation}
Z_\xi = \{x\in\Sigma^*: H(x,x)=\xi \}.
\end{equation}
We call each set $Z_\xi$ a \emph{fiber} and it is obvious that $\Sigma^*$ is a disjoint union of all $Z_\xi$, where $\xi$ runs over the range of self-similarities. For our applications, this range is finite because the sequence datasets are finite. Now consider a query $\mathscr{Q}_H(x,\kappa)$ and let $\e(x,\xi,\kappa)=(H(x,x)+\xi-2\kappa)^{1/p}$. It is easily established that
\begin{equation}\label{eq:fibers}
\mathscr{Q}_H(x,\kappa) = \bigsqcup_\xi  \cball{x}{\e(x,\xi,\kappa)} \big\vert Z_\xi,
\end{equation}
where $\cball{x}{\e(x,\xi,\kappa)}=\{y\in X: M(x,y)\leq \e(x,\xi,\kappa)\}$ (the closed ball of radius $\e(x,\xi,\kappa)$ about $x$).

Hence, to process each local similarity range query, it is sufficient to process a metric range query $\cball{x}{\e(x,\xi,\kappa)}$ on each fiber and then collect the results. For practical purposes the fibers need to be reasonably large and small in number, but that is often true because the score matrices are integer-valued. Adjacent fibers that contain too few points can be merged if care is exercised when collecting final results. Each fiber can be indexed separately as a metric space with one of the many existing access methods \cite{HjSa03} or by using a new technique. The decomposition (\ref{eq:fibers}) was proposed in $\ell^1$ form in \cite{SP07} for indexing similarity-based range queries and was in turn inspired by decomposition of weightable quasi-metric spaces into fibers used by Vitolo \cite{Vi99}.

Therefore, using fibers, a consistent indexing scheme can be constructed for most existing local similarity measures on biological sequences and profiles. The performance of such schemes is not guaranteed -- it depends on the exact geometry of sequence datasets \cite{Pe00,PS06}. Hence, our theoretical results represent only the first step towards efficient and consistent access methods that are to be achieved in future.

An alternative to fiber decomposition for cases where $\LQ{p}{\GQ{p}{s}{\gamma}{\delta}}{f}{0}$ is truly a quasi-metric is to use the quasi-metric directly to index the dataset. Pestov and Stojmirovi\'c \cite{PS06} proposed the concept of a quasi-metric tree: a general  indexing scheme for retrieving queries based on quasi-metrics and established conditions for its consistency. Note that in the $\ell^1$ case, using inequality (\ref{eqn:sim2qm011}) directly, as in \cite{IGAK05}, produces a structure that is equivalent to a quasi-metric tree.

\subsubsection*{Progressive multiple sequence alignment}\label{subsubsec:guidetree}

Multiple sequence alignment (MSA) is among the most valuable tools in computational biology. It allows extracting and representing biologically important commonalities from sets of sequences \cite{Gusfield97}. Construction of multiple alignments from sets of sequences has been extensively researched and a variety of techniques have been proposed \cite{Gusfield97,Durbin:1998}. The full dynamic programming algorithm for MSA is NP-complete \cite{WJ94} and therefore heuristics are commonly employed. One popular heuristic approach is progressive alignment \cite{FD87}. First, a guide tree is constructed from pairwise dissimilarities between sequences. Then, larger and larger groups of sequences are aligned in pairwise manner, following the branching order of the guide tree from the leaves towards the root. A number of popular software packages for MSA of protein sequences \cite{THG94,KMKM02,Edgar04,LS05} implement this heuristics.

The success of this approach, greedy in nature, crucially depends on a faithful and evolutionarily meaningful construction of a guide tree for the set of sequences to be aligned. When constructing their guide trees, most methods do not use a true metric to compute pairwise distances \cite{THG94,KMKM02,Edgar04}, while those that do \cite{LS05}, use the Levenstein distance, overlooking the similarities between closely related amino acids. 

There are advantages in using the true metric distance for agglomerative hierarchical clustering. For example, the triangle inequality ensures the transitivity of closeness in distance measure. Furthermore, when this is the case, it was shown that the difference between a hierarchical clustering and the optimal $k$-clustering is bounded \cite{DL05}. 

In this paper we have demonstrated a way to construct a large class of (quasi-)metric distances from similarity scores that also naturally account for functional relatedness among amino acids. The quasi-metrics developed in Section \ref{sec:localsim} can also provide a rigorous way to naturally interpolate from global to local similarities in constructing guide trees. 

\subsubsection*{Embeddings into vector spaces}

Let $\qsum{\sQ}$ be the metric symmetrizing the quasi-metric $\sQ=\LQ{p}{\rho}{f}{0}$, where $\qsum{Q}(x,y)=\sQ(x,y)+\sQ(y,x)$ for all $x,y\in\Sigma^*$. Observe that by the triangle inequality for $\sQ$,
\begin{equation}
(\bar{f}(x))^{1/p} - (\bar{f}(y))^{1/p}=\sQ(x,e)-\sQ(y,e)\leq \sQ(x,y),
\end{equation}
and hence, letting $\alpha(x)=(\bar{f}(x))^{1/p}$, we have
\begin{equation}\label{eqn:normedpair}
\abs{\alpha(x)-\alpha(y)}\leq \qsum{\sQ}(x,y)\leq \alpha(x)+\alpha(y).
\end{equation}

Flood, in his PhD thesis \cite{Fl75} and a followup paper \cite{Fl84} called any pair $(\rho,\alpha)$, where $\rho$ is a metric and $\alpha$ a positive function, which satisfies the above property (\ref{eqn:normedpair}), a \emph{normed pair}. The triple $(X,\rho,\alpha)$, where $(\rho,\alpha)$ is a norm pair on $X$, is called a \emph{normed set} \cite{Pestov94}. Every normed space $(E,\norm{.}_E)$ naturally becomes the normed set by setting $\rho(x,y)=\norm{x-y}_E$ and $\alpha(x)=\norm{x}_E$. 

For any two normed sets $X_1=(X_1,\rho_1,\alpha_1)$ and $X_2=(X_2,\rho_2,\alpha_2)$, a function $\pi:X_1\to X_2$ is called a contraction if for all $x\in X_1$, 
\begin{equation}
\alpha_2(\pi(x))\leq\alpha_1(x) 
\end{equation}
and for all $x,y\in X_1$,
\begin{equation}
\rho_2(\pi(x),\pi(y))\leq \rho_1(x,y).
\end{equation}
According to a result of Flood \cite{Fl75,Fl84} (see also \cite{Pestov94}), the normed pair structure supports a natural embedding of $X$ into a Banach space with a certain universal property.
\begin{thm}[\cite{Fl75,Fl84,Pestov94}]
Let $X=(X,\rho,\alpha)$ be a normed set. There exists a complete normed space $B(X)$ and an embedding of $X$ into $B(X)$ as a normed subset such that every contraction $\pi$ from $X$ to a complete normed space $E$ lifts to a unique linear contraction $\bar{\pi}\colon B(X)\to E$. The pair consisting of $B(X)$ and embedding $X\hookrightarrow B(X)$ is essentially unique. Elements of $X$ are linearly independent. \qed
\end{thm}
Therefore, spaces of biological sequences with local similarity metric may be founded upon Banach (or even Hilbert) spaces. However, this result carries only theoretical significance at this point and cannot be directly used for clustering or indexing since the free Banach space $B(X)$ is too large (it is not desirable that all sequences are linearly independent). Nevertheless, the same idea can be used to embed similarity score matrices into finite dimensional normed spaces and hence consider biological sequences as free semigroups over $\R^n$.

\section*{Acknowledgments}

A.S. is very grateful to Vladimir Pestov who as his Ph.D. and postdoctoral supervisor read and commented on the early versions of this manuscript. A.S. was supported by the University of Ottawa research funds. This work was supported by the Intramural Research Program of the National Library of Medicine at National Institutes of Health. 

%************************************************************************************
% 
% 
% % The Appendices part is started with the command \appendix;
% % appendix sections are then done as normal sections
\appendix

\newpage
\section{Proofs}\label{app:proofs}

\subsection{General conditions for edit quasi-metrics}\label{app:editdist}

\begin{thm}\label{thm:editdistqm1}
Let $\Sigma$ be a set, let $1\leq p<\infty$ and suppose $d$ is a separating quasi-metric on $\Sigma$, $\alpha,\beta\in\Gamma$ and $\sD$ is the $\ell^p$ edit distance  extending $d$,$\alpha$ and $\beta$. In addition, assume that for all $a,b\in\Sigma$, $u,v,x\in\Sigma^*$,
\begin{itemize}
\item[(W1)] $\displaystyle d^p(a,b)+\beta^p(ubv) \geq \beta^p(uav)$;
\item[(W2)] $\displaystyle d^p(a,b)+\alpha^p(uav) \geq \alpha^p(ubv)$;
\item[(W3)] $\displaystyle \beta^p(uv)+ \beta^p(x)\geq \beta^p(uxv)$;
\item[(W4)] $\displaystyle \alpha^p(uv)+ \alpha^p(x)\geq \alpha^p(uxv)$;
\item[(W5)] $\displaystyle \beta^p(uxv)+\alpha^p(x)\geq \beta^p(uv)$;
\item[(W6)] $\displaystyle \alpha^p(uxv)+ \beta^p(x)\geq \alpha^p(uv)$;
\item[(W7)] $\displaystyle \alpha^p(ux)+ \beta^p(xv)\geq \alpha^p(u)+ \beta^p(v)$;
\item[(W8)] $\displaystyle \beta^p(ux)+\alpha^p(xv)\geq \beta^p(u)+ \alpha^p(v)$.
\end{itemize}
Then, $\sD$ is a separating quasi-metric on $\Sigma^*$.
\end{thm}

\begin{proof}
Let $x,y,z\in\Sigma^*$. Clearly, $\sD(x,y)$ is non-negative since all of $d$, $\alpha$ and $\beta$ are non-negative. Also, $\sD(x,x)\leq \left(\sum_i d^p(x_i,x_i)\right)^{1/p} = 0$. Now suppose $\sD(x,y)=0$. Applying Lemma \ref{lemma:aligndecomp}, we have
\begin{equation*}
\sD(x,y) = \left(\sum_{k=1}^K \sD^p(x^*_k,y^*_k)\right)^{1/p}=0,
\end{equation*}
where $x=x^*_1x^*_2\ldots x^*_K$, $y=y^*_1y^*_2\ldots y^*_K$, implying $\sD(x^*_k,y^*_k)=0$ for all $k$ since $\sD$ is non-negative. Hence, $x^*_k=y^*_k$ for all possible cases of $x^*_k$ and $y^*_k$ because $d$ is a separating quasi-metric and $\alpha$ and $\beta$ are strictly positive on $\Sigma^+$.

We will demonstrate the triangle inequality by relying on the Minkowski inequality:
for any two sequences $a$ and $b$ of real numbers and $1\leq p<\infty$,
\begin{equation}
\left(\sum_i \abs{a_i+b_i}^p\right)^{1/p}\leq \left(\sum_i \abs{a_i}^p\right)^{1/p} + \left(\sum_i \abs{b_i}^p\right)^{1/p}.
\end{equation}
We show by induction that for all $0\leq i\leq\abs{x}$, $0\leq j\leq\abs{y}$ and $0\leq k\leq\abs{z}$,
\begin{equation}
\sD(\bar{x}_i,\bar{y}_j) + \sD(\bar{y}_j,\bar{z}_k)\geq \sD(\bar{x}_i,\bar{z}_k).
\end{equation}
Let $\preceq$ denote a partial order on $\N\times\N\times\N$ where $(i_0,j_0,k_0)\preceq (i,j,k)$ if $i_0\leq i$ or $i_0=i$ and $j_0\leq j$ or $i_0=i$ and $j_0=j$ and $k_0\leq k$ (lexicographic order). The relation $\preceq$ is a well-founded partial order of type $\omega^3$ (in this case our induction is finite) and our claim is trivially true for $(0,0)$. Assume it is true for all $(i',j',k')\prec (i,j,k)$. There are nine possibilities in total to consider for $(i',j',k')= (i,j,k)$.

\textbf{Case 1:} Suppose $\sD(\bar{x}_i,\bar{y}_j)= \left(\sD^p(\bar{x}_{i-1},\bar{y}_{j-1})+ d^p(x_i,y_j)\right)^{1/p}$ and $\sD(\bar{y}_j,\bar{z}_k)= (\sD^p(\bar{y}_{j-1},\bar{z}_{k-1})+d^p(y_j,z_k))^{1/p}$. By the Minkowski inequality, our induction hypothesis and the triangle inequality on $d$ we have
\begin{align*}
\sD(\bar{x}_i,\bar{y}_j) + \sD(\bar{y}_j,\bar{z}_k) &= \phantom{+} \left(\sD^p(\bar{x}_{i-1},\bar{y}_{j-1})+ d^p(x_i,y_j)\right)^{1/p}\\ &\phantom{=}\  + (\sD^p(\bar{y}_{j-1},\bar{z}_{k-1})+d^p(y_j,z_k))^{1/p}\\
& \geq \big( (\sD(\bar{x}_{i-1},\bar{y}_{j-1})+  \sD(\bar{y}_{j-1},\bar{z}_{k-1}) )^p\\ & \phantom{=} \ + (d(x_i,y_j)+d(y_j,z_k))^p \big)^{1/p}\\
& \geq \big(\sD^p(\bar{x}_{i-1},\bar{z}_{k-1}) + d^p(x_i,z_k)\big)^{1/p}\\
& \geq \sD(\bar{x}_i,\bar{z}_k).
\end{align*}

\textbf{Case 2:} Suppose $\sD(\bar{y}_j,\bar{z}_k)= (\sD^p(\bar{y}_{j},\bar{z}_{k-t})+\alpha^p(z_{k-t+1}\ldots z_k))^{1/p}$ for some $1\leq t< k$ (this covers three possibilities). By the Minkowski inequality and the induction hypothesis we have
\begin{align*}
\sD(\bar{x}_i,\bar{y}_j) + \sD(\bar{y}_j,\bar{z}_k) &= \sD(\bar{x}_i,\bar{y}_j) + (\sD^p(\bar{y}_{j},\bar{z}_{k-t})+ \alpha^p(z_{k-t+1}\ldots z_k))^{1/p}\\
& \geq \big( (\sD(\bar{x}_i,\bar{y}_j)+\sD(\bar{y}_{j},\bar{z}_{k-t}))^p +  \alpha^p(z_{k-t+1}\ldots z_k) \big)^{1/p}\\
& \geq \big(\sD^p(\bar{x}_{i},\bar{z}_{k-t}) + \alpha^p(z_{k-t+1}\ldots z_k) \big)^{1/p}\\
& \geq \sD(\bar{x}_i,\bar{z}_k).
\end{align*}

\textbf{Case 3:} Suppose $\sD(\bar{x}_i,\bar{y}_j)= (\sD^p(\bar{x}_{i-t},\bar{y}_{j})+ \beta^p(x_{i-t+1}\ldots x_i))^{1/p}$ for some $1\leq t< i$ (this covers additional two possibilities). Then, in similar manner as in Case 2,
\begin{equation*}
\sD(\bar{x}_i,\bar{y}_j) + \sD(\bar{y}_j,\bar{z}_k) \geq \big( \sD^p(\bar{x}_{i-t},\bar{z}_{k}) +\beta^p(x_{i-t+1}\ldots x_i) \big)^{1/p} \geq \sD(\bar{x}_i,\bar{z}_k),
\end{equation*}
by the Minkowski inequality and the induction hypothesis. 

\textbf{Case 4:} Suppose $\sD(\bar{y}_j,\bar{z}_k)= \big(\sD^p(\bar{y}_{j-t},\bar{z}_{k})+ \beta^p(y_{j-t+1}\ldots y_j)\big)^{1/p}$, for some $1\leq t< j$ (this covers additional two possibilities). Using Lemma \ref{lemma:aligndecomp}, let $0\leq q\leq j$ be the smallest integer not larger than $t$ such that 
\begin{equation*}
\sD(\bar{x}_i,\bar{y}_j) = \left( \sD^p(\bar{x}_r,\bar{y}_{j-q}) + \sum_{m=1}^K \sD^p(u^*_m,v^*_m) \right)^{1/p},
\end{equation*}
for some $1\leq r\leq i$, where $u=x_{r+1}\ldots x_i=u^*_1\ldots u^*_K$, and $v=y_{j-q+1}\ldots y_j=v^*_1\ldots v^*_K$. Note that $q<t$ if and only if $\sD(\bar{x}_r,\bar{y}_{j-q})= \big(\sD^p(\bar{x}_r,\bar{y}_{j-q'})+\sD^p(e, y_{j-q'+1}\ldots y_{j-q})\big)^{1/p}$, where $q<t<q'\leq j$. In that case, by our assumption (W7) and by Minkowski inequality,
\begin{align*}
\sD^p(\bar{x}_r,\bar{y}_{j-q}) +  \beta^p(y_{j-t+1}\ldots y_j) & = \sD^p(\bar{x}_r,\bar{y}_{j-q'}) + \alpha^p(y_{j-q'+1}\ldots y_{j-q}) \\
& \phantom{\ }  +\beta^p(y_{j-t+1}\ldots y_j)\\
& \geq  \sD^p(\bar{x}_r,\bar{y}_{j-q'}) + \alpha^p(y_{j-q'+1}\ldots y_{j-t})\\
& \phantom{\ } + \beta^p(y_{j-q+1}\ldots y_j) \\
& \geq \sD^p(\bar{x}_r,\bar{y}_{j-t}) + \beta^p(v).
\end{align*}
Of course, the same inequality trivially holds if $t=q$.

Observe that assumptions (W1), (W3) and (W5) imply that for any $1\leq m\leq K$ and any $w_1,w_2\in\Sigma^*$,
\begin{equation}
\sD^p(u^*_m,v^*_m) + \beta^p(w_{1}v^*_{m}w_2)\geq \beta^p(w_{1}u^*_{m}w_{2}),
\end{equation}
and hence
\begin{align*}
\sum_{m=1}^K \sD^p(u^*_m,v^*_m) + \beta^p(y_{j-q+1}\ldots y_j) & \geq \sum_{m=2}^K \sD^p(u^*_m,v^*_m) + \beta^p(u^*_1v^*_2\ldots v^*_K)\\
& \geq \sum_{m=3}^K \sD^p(u^*_m,v^*_m) + \beta^p(u^*_1u^*_2v^*_3\ldots v^*_K)\\
& \geq \beta^p(u^*_1\ldots u^*_K)\\
& = \beta^p(u).
\end{align*}

Therefore, 
\begin{align*}
\sD(\bar{x}_i,\bar{y}_j) + \sD(\bar{y}_j,\bar{z}_k) &=
\left( \sD^p(\bar{x}_r,\bar{y}_{j-q}) + \sum_{m=1}^K \sD^p(u^*_m,v^*_m) \right)^{1/p} \\
& \phantom{\geq} + \Big(\sD^p(\bar{y}_{j-t},\bar{z}_{k})+ \beta^p(y_{j-t+1}\ldots y_j)\Big)^{1/p}\\
& \geq \Bigg( \sD^p(\bar{x}_r,\bar{y}_{j-q}) + \sum_{m=1}^K \sD^p(u^*_m,v^*_m) + \sD^p(\bar{y}_{j-t},\bar{z}_{k})\\
& \phantom{\geq\quad} + \beta^p(y_{j-t+1}\ldots y_j)\Bigg)^{1/p}\\
& \geq   \Bigg( \sD^p(\bar{x}_r,\bar{y}_{j-t}) + \sum_{m=1}^K \sD^p(u^*_m,v^*_m) + \sD^p(\bar{y}_{j-t},\bar{z}_{k})\\
& \phantom{\geq\quad} + \beta^p(y_{j-q+1}\ldots y_j)\Bigg)^{1/p}\\
& \geq \Big(\sD^p(\bar{x}_r,\bar{z}_k) + \beta^p(u)\Big)^{1/p}\\
& \geq \sD(\bar{x}_i,\bar{z}_k),
\end{align*}
by the induction hypothesis.

\textbf{Case 5:} The remaining case is $\sD(\bar{x}_i,\bar{y}_j)= (\sD^p(\bar{x}_{i},\bar{y}_{j-t})+ \alpha^p(y_{j-t+1}\ldots y_j))^{1/p}$ for some $1\leq t< j$ and $\sD(\bar{y}_j,\bar{z}_k)= (\sD^p(\bar{y}_{j-1},\bar{z}_{k-1}) +d^p(y_j,z_k))^{1/p}$. The proof for this case exactly mirrors the proof for the previous case, now depending on the assumptions (W2), (W4), (W6) and (W8).
\end{proof}

\begin{rem}
In general the assumptions (W1) -- (W8) are sufficient for $\sD$ to be a quasi-metric but not necessary, except in the case of $p=1$. For example, let $\Sigma=\{a,b\}$, $d(a,b)=d(b,a)=3$, $\alpha=\beta$, $\alpha(a)=7$, $\alpha(b)=4$, $\alpha(u)=\sum_i \alpha(u_i)$. In this case the assumptions (W1) and (W2) fail but it can be verified that the triangle inequality for $\sD$ does not fail for any $p>1$.
\end{rem}

\begin{rem}
The assumptions (W1)--(W8) can be significantly simplified if the gap penalties take a more restricted form. For example, if the gap penalties are increasing, the assumptions (W5)--(W8) can be removed. This restriction is sensible in applications to biological sequence comparisons because algebraic interactions lowering the effective length of the sequence are not allowed. On the other hand, if $\Sigma^*$ is replaced as the underlying set with a monoid which is not free, or even a group, then gap penalties cannot be increasing in the above sense.

Since composition-length gap penalties are increasing by definition, Theorem \ref{thm:editdistqm} is a direct corollary of Theorem \ref{thm:editdistqm1}. Furthermore, 
composition-length gap penalties with $\phi=0$, such as linear or affine, satisfy all of (W1)--(W8).
\end{rem}

\subsection{Global similarities}\label{app:prflemmaselfsim}

\setcounter{tmpthm}{\value{thm}} 
\setcounter{tmpsec}{\value{section}}

\setcounter{thm}{\value{lemselfsimn}} 
\setcounter{section}{\value{lemselfsims}}
\renewcommand{\thesection}{\arabic{section}}

\begin{prop}
Let $\Sigma$ be a set and let $s:\Sigma\times\Sigma\to\R$ be a a sane scoring function over $\Sigma$. Suppose $\gamma,\delta\in\Gamma(\Sigma)$ and $\sS$ the global similarity on $\Sigma^*$ with respect to $s,\delta$ and $\gamma$. Then, $\sS$ is a sane scoring function and for all $x\in\Sigma^*$,  
\begin{equation}
\sS(x,x) = \sum_{i=1}^{\abs{x}} s(x_i,x_i).
\end{equation}
\qed
\end{prop}

\setcounter{thm}{\value{tmpthm}} 
\setcounter{section}{\value{tmpsec}}
\renewcommand{\thesection}{\Alph{section}}

We will make use of the following lemma, equivalent to Lemma \ref{lemma:aligndecomp} for distances. It was likewise proved by Smith and Waterman \cite{SW81a} for the $\ell^1$ case and less general gap penalties.

\begin{lem}\label{lemma:selfsim2}
Let $\Sigma$ be a set, $s:\Sigma\times\Sigma\to\R$, and $\gamma,\delta:\Sigma^+\to\R_+$. Suppose $\sS$ is a global similarity on $\Sigma^*$ with respect to $d$, $\gamma$ and $\delta$. Then, for all $x,y\in\Sigma^*$
\begin{equation}
\sS(x,y) = \max \biggl\{ {\textstyle \sum_{k=1}^K \sS(x^*_k,y^*_k)} \  \big\vert\ \algn{(x^*_k,y^*_k)}_{k=1}^K\in\mathcal{A}(x,y) \biggr\}.
\end{equation}
\end{lem}

\begin{proof}[Proof of Proposition \ref{lemma:selfsim}]
Let $x,y\in\Sigma^*$. If $x=e$, by definition $\sS(x,x)=0$, coinciding with a sum over the empty set. Since $\gamma$ and $\delta$ are positive, we have $-\gamma(y)=S(e,y)\leq 0$ and $-\delta(y)=S(y,e)\leq 0$.

Now suppose $x\in\Sigma^+$ and let $\algn{(x^*_k,y^*_k)}_{k=1}^K\in\mathcal{A}(x,y)$ such that $\sS(x,y)=\sum_{k=1}^K \sS(x^*_k,y^*_k)$. Let $C=\{k:x^*_k\in\Sigma\ \text{and}\ y^*_k\in\Sigma \}$ and $D=\{k:x^*_k\in\Sigma^+\ \text{and}\ y^*_k=e \}$. Then, 
\begin{align*}
\sS(x,y) & \leq \sum_{k\in C} \sS(x^*_k,y^*_k) + \sum_{k\in D} \sS(x^*_k,y^*_k)\\
& \leq \sum_{k\in C} s(x^*_k,y^*_k) - \sum_{k\in D} \delta(x^*_k)\\
& \leq \sum_{k\in C} s(x^*_k,x^*_k) + \sum_{k\in D}\sum_j s((x^*_k)_j,(x^*_k)_j)\\
& = \sum_{i=1}^{\abs{x}} s(x_i,x_i),
\end{align*}
since $s$ is sane and the whole of $x$ is accounted for in fragments indexed by $C$ and $D$. Therefore,
\begin{equation}
\sS(x,y) \leq \sum_{i=1}^{\abs{x}} s(x_i,x_i) \leq \sS(x,x),
\end{equation}
implying $\sS(x,x)=\sum_{i=1}^{\abs{x}} s(x_i,x_i)>0$ and $\sS(x,x)\geq\sS(x,y)$. In the same way it can be shown that $\sS(x,x)\geq\sS(y,x)$ and hence that $\sS$ is sane.
\end{proof}

\setcounter{tmpthm}{\value{thm}} 
\setcounter{tmpsec}{\value{section}}
\setcounter{thm}{\value{corglobdistn}} 
\setcounter{section}{\value{corglobdists}}
\renewcommand{\thesection}{\arabic{section}}
\begin{cor}%\label{cor:simdist2}
Let $\Sigma$ be a set and let $1\leq p<\infty$. Suppose $s$ is a sane scoring function on $\Sigma$, $d=\AQ{p}{s}$ is a quasi-metric on $\Sigma$ and $\gamma,\delta\in\gmcl$ such that
\begin{equation}
\gamma(b)-\gamma(a) \leq d^p(a,b)
\end{equation}
and
\begin{equation}%\label{eq:gapcond2}
s(a,a)+\delta(a)-s(b,b)-\delta(b) \leq d^p(a,b).
\end{equation}
Let $\sS$ be the global similarity with respect to $s,\gamma$ and $\delta$ and let $\alpha(x)=\gamma(x)^{1/p}$ and $\beta(x)=\bigl(\sS(x,x)+\delta(x)\bigr)^{1/p}$ for all $x\in\Sigma^+$. Then, the $\ell^p$ edit distance $\sD=\DQ{p}{d}{\alpha}{\beta}$ is given for all $x,y\in\Sigma^*$ by the formula 
\begin{equation}\label{eqn:qm_new}
\sD(x,y) = \Bigl(\sS(x,x) - \sS(x,y)\Bigr)^{1/p}.
\end{equation}
\end{cor}
\begin{proof}
By construction, $\alpha^p\in\gmcl$ and by Proposition \ref{lemma:selfsim}, $\beta^p\in\gmcl$ as well. By our assumptions on $d$, $\gamma$ and $\delta$ and by Theorem \ref{thm:editdistqm}, it follows that $\sD$, the $\ell^p$ edit distance extending $d$, $\alpha$ and $\beta$, is indeed the separating quasi-metric $\DQ{p}{d}{\alpha}{\beta}$ on $\Sigma^*$. We will now show by recursion that this quasi-metric is equivalent to the one given by Equation (\ref{eqn:qm_new}).

Clearly, $\sD(e,e)=\bigl(\sS(e,e) - \sS(e,e)\bigr)^{1/p}=0$. Let $x,y\in\Sigma^+$ and suppose $1\leq i\leq\abs{x}$ and $1\leq j\leq\abs{y}$. We have, \[\sD(e,\bar{y}_j)=\alpha(\bar{y}_j)=\gamma(\bar{y}_j)^{1/p}=\bigl(\sS(e,e) - \sS(e,\bar{y}_j)\bigr)^{1/p},\] and
\[\sD(\bar{x}_i,e)=\beta(\bar{x}_i)=\bigl(\delta(\bar{x}_i) + \sS(\bar{x}_i,\bar{x}_i) \bigr)^{1/p} =\bigl(\sS(\bar{x}_i,\bar{x}_i) - \sS(\bar{x}_i,e)\bigr)^{1/p}.\]
Using recursion and Proposition \ref{lemma:selfsim},

\begin{align*}
%\sD(x,y) & = \left(\sum_{k=1}^K \sD^p(x^*_k,y^*_k)\right)^{1/p}\\
%& = \Big(\sD^p(\check{x},\check{y})+ \sD^p(x^*_m,uu') + \sum_{k=m+1}^{n-1} \sD^p(x^*_k,y^*_k) + \sD^p(x^*_n,v'v) + \sD^p(\hat{x},\hat{y})\Big)^{1/p}\\
\sD(\bar{x}_i,\bar{y}_j) &=\bigg(  \min\Big\lbrace \sD^p(\bar{x}_{i-1},\bar{y}_{j-1})+d^p(x_i,y_j), \\
&\qquad\qquad\quad \min_{1\leq k\leq j} \left\{\sD^p(\bar{x}_{i},\bar{y}_{j-k})+\alpha^p(y_{j-k+1}\ldots y_j)\right\},\\ 
&\qquad\qquad\quad \min_{1\leq k\leq i}\left\{\sD^p(\bar{x}_{i-k},\bar{y}_{j})+ \beta^p(x_{i-k+1}\ldots x_i)\right\} \Big\rbrace \bigg) ^{1/p}\\
& = \bigg(  \min\Big\lbrace \sS(\bar{x}_{i-1},\bar{x}_{i-1})- \sS(\bar{x}_{i-1},\bar{y}_{j-1}) +s(x_i,x_i) - s(x_i,y_j), \\
&\qquad\qquad\quad \min_{1\leq k\leq j}\left\{ \sS(\bar{x}_{i},\bar{x}_{i}) - \sS(\bar{x}_{i},\bar{y}_{j-k}) + \gamma(y_{j-k+1}\ldots y_j)\right\},\\ 
&\qquad\qquad\quad \min_{1\leq k\leq i} \big\{ \sS(\bar{x}_{i-k},\bar{x}_{i-k}) -\sS(\bar{x}_{i-k},\bar{y}_{j}) + \delta(x_{i-k+1}\ldots x_i) \\
&\qquad\qquad\quad \qquad\quad +\sS(x_{i-k+1}\ldots x_i,x_{i-k+1}\ldots x_i)\big\} \Big\rbrace \bigg) ^{1/p}\\
& =\bigg( \sS(\bar{x}_{i},\bar{x}_{i}) -\max\Big\lbrace \sS(\bar{x}_{i-1},\bar{y}_{j-1})+s(x_i,y_j), \\
&\qquad\qquad\quad \qquad\qquad\;  \max_{1\leq k\leq j}\left\{\sS(\bar{x}_{i},\bar{y}_{j-k})-\gamma(y_{j-k+1}\ldots y_j)\right\},\\ 
&\qquad\qquad\quad \qquad\qquad\; \max_{1\leq k\leq i}\left\{\sS(\bar{x}_{i-k},\bar{y}_{j})-\delta(x_{i-k+1}\ldots x_i)\right\} \Big\rbrace\bigg) ^{1/p} \\
& = \Bigl( \sS(\bar{x}_{i},\bar{x}_{i}) - \sS(\bar{x}_{i},\bar{y}_{j}) \Bigr) ^{1/p},
\end{align*}
as required.
\end{proof}
\setcounter{thm}{\value{tmpthm}} 
\setcounter{section}{\value{tmpsec}}
\renewcommand{\thesection}{\Alph{section}}

% 
% ********************* BIBTEX BIBLIOGRAPHY **************************
% \bibliographystyle{abbrv}
% \bibliography{mathbib,compbib,biobib}
% ********************************************************************
%

\def\cdprime{$''$}

\end{document}